\pdfoutput=1
\documentclass[]{sigplanconf}

\makeatletter
\def \ps@plain {%
  \let \@mkboth = \@gobbletwo
  \let \@evenhead = \@empty
  \def \@evenfoot {\scriptsize \hfil \thepage \hfil}%
  \let \@oddhead = \@empty
  \let \@oddfoot = \@evenfoot}
\makeatother

\usepackage[utf8]{inputenc}
\usepackage[T1]{fontenc}

\usepackage{balance}
\usepackage{amsmath}
\usepackage{amsthm}
\usepackage{upgreek}
\usepackage{amssymb}
\usepackage{graphics}
\usepackage[all]{xy}
\usepackage{textcomp}
\usepackage{color}
\usepackage{amscd}
\usepackage{kbordermatrix}
\usepackage{fancyvrb}
\usepackage{float}
\usepackage{etex} 
\usepackage{tikz}
\usetikzlibrary{matrix,arrows,positioning,chains,scopes,fit,shapes,decorations,decorations.pathreplacing,calc,trees,patterns}
\usepackage{pgfplots}
\pgfplotsset{compat=1.3}
\usepackage{tikz-qtree}

\newcounter{defs}
\newtheorem{theorem}[defs]{Theorem}
\newtheorem{lemma}[defs]{Lemma}
\newtheorem{assumption}[defs]{Assumption}
\newtheorem{fallacy}[defs]{Fallacy}

\newcommand{\etal}{\emph{et al.}}

\newcommand{\syn}[1]{\mathsf{#1}}
\newcommand{\var}[1]{\mathit{#1}}

\newcommand{\parto}{\rightharpoonup}

\newcommand{\set}[1]{\{#1\}}

\newcommand{\PowSm}[1]{{\mathcal{P}(#1)}}

\newcommand{\join}{\sqcup}

\newcommand{\opor}{\mathrel{|}}

\newcommand{\produces}{\mathrel{::=}}

\newcommand{\vv}{x}

\newcommand{\lam}{\ensuremath{\var{lam}}}

\newcommand{\ttlp}{\mbox{\tt (}}
\newcommand{\ttrp}{\mbox{\tt )}}

\newcommand{\appform}[2]{\ttlp #1\; #2\ttrp}
\newcommand{\lamform}[2]{\ttlp \uplambda\;\ttlp#1\ttrp\;#2\ttrp}

\newcommand{\letiform}[3]{\ttlp {\tt let}\; \ttlp\texttt{\lbrack}#1\; #2\texttt{\rbrack}\ttrp\; #3\ttrp}

\newcommand{\fexpr}{f}
\newcommand{\expr}{e}
\newcommand{\aexpr}{\mbox{\sl {\ae}}}

\newcommand{\State}{\Sigma}
\newcommand{\state}{\varsigma}

\newcommand{\store}{\sigma}

\newcommand{\env}{\rho}

\newcommand{\clo}{\var{clo}}
\newcommand{\cont}{\kappa}

\newcommand{\addr}{a}

\newcommand{\sa}[1]{\widehat{\mathit{#1}}}

\newcommand{\aArgEval}{{\hat{\mathcal{A}}}}

\newcommand{\aState}{{\hat{\Sigma}}}
\newcommand{\astate}{{\hat{\varsigma}}}

\newcommand{\astore}{{\hat{\sigma}}}
\newcommand{\aenv}{{\hat{\rho}}}

\newcommand{\aclo}{{\widehat{\var{clo}}}}

\newcommand{\acont}{{\hat{\kappa}}}

\newcommand{\aaddr}{{\hat{\addr}}}
\newcommand{\aalloc}{{\sa{alloc}}}

\newcommand{\defas}{\triangleq}

\newcommand{\wf}{\mathit{wf}}

\newcommand{\cTo}{\leadsto_{\scalebox{0.45}{$\Sigma$}}}
\newcommand{\cToW}{\leadsto_{\scalebox{0.45}{$S$}}}
\newcommand{\fsTo}{\leadsto^{\text{\texttildelow}}_{\scalebox{0.45}{$\Sigma$}}}
\newcommand{\pdTo}{\leadsto^{\scalebox{0.45}[0.25]{$\wedge$}}_{\scalebox{0.45}{$\Sigma$}}}
\newcommand{\fsToW}{\leadsto^{\text{\texttildelow}}_{\scalebox{0.45}{$\Xi$}}}
\newcommand{\pdToW}{\leadsto^{\scalebox{0.45}[0.25]{$\wedge$}}_{\scalebox{0.45}{$\Xi$}}}
\newcommand{\fsToS}{\leadsto^{\text{\texttildelow}}_{\scalebox{0.45}{$S$}}}
\newcommand{\pdToS}{\leadsto^{\scalebox{0.45}[0.25]{$\wedge$}}_{\scalebox{0.45}{$S$}}}

\newcommand{\pdToSub}{\overset{\scalebox{0.45}[0.45]{$\!\!\!\!\!\sqsubseteq$}}{\leadsto^{\scalebox{0.45}[0.25]{$\wedge$}}_{\scalebox{0.45}{$\Sigma$}}}}
\newcommand{\fsToSub}{\overset{\scalebox{0.45}[0.45]{$\!\!\!\!\!\sqsubseteq$}}{\leadsto^{\text{\texttildelow}}_{\scalebox{0.45}{$\Sigma$}}}}
\newcommand{\fsPath}{\overset{\!\scalebox{0.55}[0.55]{\texttildelow}}{\hookrightarrow}}
\newcommand{\pdPath}{\overset{\!\scalebox{0.45}[0.25]{$\wedge$}}{\hookrightarrow}}

\newcommand{\simXi}{\sqsupseteq_{\scalebox{0.5}{$\Xi$}}}

\newcommand{\simR}{\sqsupseteq_{\scalebox{0.45}{$R$}}}

\newcommand{\simEnv}{\equiv_{\scalebox{0.45}{$Env$}}}

\newcommand{\simStore}{\sqsupseteq_{\scalebox{0.45}{$Store$}}}
\newcommand{\simD}{\sqsupseteq_{\scalebox{0.45}{$D$}}}
\newcommand{\simClo}{\equiv_{\scalebox{0.45}{$Clo$}}}

\newcommand{\simFrame}{\equiv_{\scalebox{0.45}{$Frame$}}}

\newcommand{\tEnv}{T_{\scalebox{0.45}[0.45]{$Env$}}}
\newcommand{\tClo}{T_{\scalebox{0.45}[0.45]{$Clo$}}}
\newcommand{\tD}{T_{\scalebox{0.45}[0.45]{$D$}}}
\newcommand{\tStore}{T_{\scalebox{0.45}[0.45]{$Store$}}}
\newcommand{\tFrame}{T_{\scalebox{0.45}[0.45]{$Frame$}}}

\newcommand{\hEnv}{H_{\scalebox{0.45}[0.45]{$Env$}}}
\newcommand{\hClo}{H_{\scalebox{0.45}[0.45]{$Clo$}}}
\newcommand{\hD}{H_{\scalebox{0.45}[0.45]{$D$}}}
\newcommand{\hStore}{H_{\scalebox{0.45}[0.45]{$Store$}}}
\newcommand{\hFrame}{H_{\scalebox{0.45}[0.45]{$Frame$}}}
\newcommand{\hKont}{H_{\scalebox{0.45}[0.45]{$Kont$}}}
\newcommand{\hC}{H_{\scalebox{0.45}[0.45]{$C$}}}

\newcommand{\inpsi}{\in_{\scalebox{0.7}[0.7]{$\psi$}}}

\usepackage[firstpage]{draftwatermark}
\SetWatermarkText{\hspace*{7.9in}\raisebox{7.05in}{\includegraphics[scale=0.1]{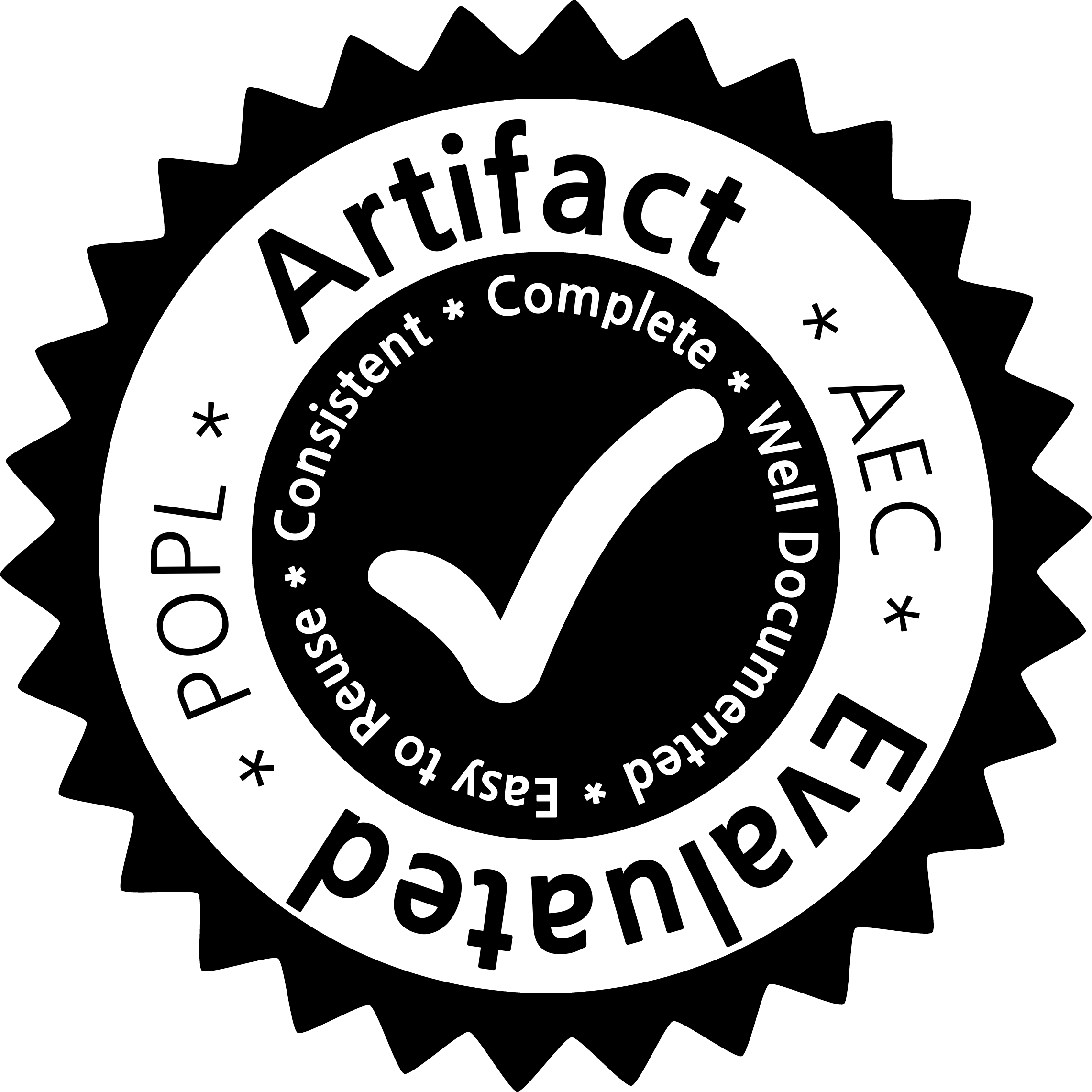}}}
\SetWatermarkAngle{0}

\begin{document}

\setlength{\pdfpageheight}{\paperheight}
\setlength{\pdfpagewidth}{\paperwidth}

\conferenceinfo{PoPL '16}{January 20--22, 2016, St. Petersburg, FL, USA} 
\copyrightyear{2016}
\copyrightdata{}
\thedoi{2837614.2837631}




\authorversion{2016}{\textit{POPL ’16: Proceedings of the 43rd annual ACM SIGPLAN Symposium on Principles of Programming Languages}}{January 2016}{2837614.2837631}


\title{Pushdown Control-Flow Analysis for Free}


\authorinfo{\iffalse\color{white}\fi Thomas Gilray \and Steven Lyde \and Michael D. Adams \and Matthew Might \and David Van Horn}
           {University of Utah, USA \and\and University of Maryland, USA}
           {\{tgilray,lyde,adamsmd,might\}@cs.utah.edu,\ \ dvanhorn@cs.umd.edu}

\maketitle

\begin{abstract}
Traditional control-flow analysis (CFA) for higher-order languages
introduces spurious connections between callers and callees, and
different invocations of a function may pollute each other's return flows. 
Recently, three distinct approaches have been published that provide perfect call-stack precision in a computable manner: CFA2, PDCFA, and AAC.
Unfortunately, implementing CFA2 and \linebreak PDCFA requires significant engineering effort.
Furthermore, all three are computationally expensive.
For a monovariant analysis, CFA2 is in $O(2^n)$, PDCFA is in $O(n^6)$, and AAC is in $O(n^8)$.

In this paper, we describe a new technique that builds on these but is both straightforward to implement and computationally inexpensive.
The crucial insight is an unusual state-dependent allocation strategy for the addresses of continuations.
Our technique imposes only a constant-factor overhead on the underlying analysis and
costs only $O(n^3)$ in the monovariant case.
We present the intuitions behind this development, benchmarks demonstrating its efficacy, and a proof of the precision of this analysis.
\end{abstract}




\section{Introduction}
\label{sec:intro}
Recent developments in the static analysis of higher-order languages make it possible to obtain perfect precision in modeling the call stack.
This allows calls and returns to be matched up precisely and avoids spurious return flows.
Consider the following Racket code, which binds an identity function and applies it on two distinct values:
\begin{Verbatim}[commandchars=\\\{\}]
    (let* ([id (lambda (x) x)]
           [y (id #t)]
           [z (id #f)])
      ...)
\end{Verbatim}
Without a precise modeling of the call stack, the value \texttt{\#f} can  spuriously flow to the variable \texttt{y}, even when a technique like
call sensitivity initially keeps them separate.

To avoid this imprecision, Vardoulakis and Shivers \cite{vardoulakis:2010:cfa2} introduces a \textit{context-free approach} (as in context-free languages, not context sensitivity) to 
program analysis with CFA2.
This technique provides a computable, although exponential-time, method for obtaining perfect stack precision for monovariant analyses of continuation-passing-style programs.
Two other approaches, PDCFA and AAC, build on this work by enabling polyvariant (e.g., context sensitive) analysis of direct-style programs and do so at only a polynomial-factor 
increase to the run-time complexity of the underlying analysis.

Earl~\etal~\cite{earl:2010:pushdown} presents a \textit{pushdown control-flow analysis} (PDCFA), which improves on traditional control-flow analysis by annotating edges 
in the state graph with stack actions (i.e., push and pop) that implicitly represent precise call stacks.
But, this method obtains its precision at a substantial increase in worst-case complexity.
For example, a monovariant PDCFA is in $O(n^6)$ where its finite-state equivalent is in $O(n^3)$.
Unfortunately, PDCFA also requires significant machinery and presents challenges to engineers responsible for constructing and maintaining such analyses.

Johnson and Van Horn~\cite{johnson:2014:aac} presents \textit{abstracting abstract control} (AAC), a refinement of store-allocated continuations with the established finite-state 
method of merging stack frames into the store, and defines an allocator that is precise enough to avoid all spurious merging.
The key advantage of this method is that it is trivial to implement in existing analysis frameworks that use store-allocated continuations and comes at the cost of changing roughly one line of code.
Unfortunately, AAC is more computationally complex than PDCFA as even in the monovariant case it is in $O(n^8)$.

We draw on the lessons learned from all three approaches and present a technique for obtaining perfect call-stack precision at only a constant-factor increase to run-time complexity over traditional finite-state analysis
(i.e., \textit{for free} in terms of complexity) and requiring no refactoring of analyses already using store-allocated continuations (i.e., \textit{for free} in terms of labor).

\subsection{Contributions}
\label{sec:intro:contributions}
We contribute an efficient method for obtaining a perfectly precise modeling of the call stack in static analyses. Specifically:
\begin{itemize}
\item
We present a novel technique for obtaining perfect call-stack precision at no asymptotic cost to run-time complexity and requiring only a
trivial change to analyses already using store-allocated continuations.
In the monovariant case, our analysis is in $O(n^3)$, the same complexity class as a traditional $0$-CFA. 
\item
We illustrate the intuition behind our approach and explain why previous PTIME methods (PDCFA and AAC) fail to exploit it.
\item
We describe our implementation and provide benchmarks that demonstrate its efficacy.
\item
We define a relationship between our technique and a static analysis that uses unbounded stacks and use it to prove the precision of our method.
\end{itemize}

\subsection{Outline}
\label{sec:intro:outline}
Section~\ref{sec:background} defines a simple direct-style language and its operational semantics. It presents the relevant background on abstract interpretation 
using abstract machines, soundness, store widening, and concepts necessary to understanding our technique.
We close this section by giving a walkthrough of the above example, illustrating precisely how values become merged in a traditional analysis.

In section~\ref{sec:perfect}, we formalize an incomputable static analysis that defines what is meant by perfect stack precision.
This analysis loses no precision in its modeling of the call stack but requires an infinite number of unbounded stacks to be explored.
We then review PDCFA and AAC, the existing polynomial-time approaches to obtaining an equivalent stack precision.

In section~\ref{sec:forfree}, we formalize our technique, give the intuitions that led us to it, and explain how it relates to each of the analyses 
described in section~\ref{sec:perfect}.
We describe our implementation and present both monovariant and call-sensitive allocation benchmark results that compare the complexity and 
precision of our technique to that of ACC.

Section~\ref{sec:proof} provides a formal relationship between the unbounded-stack machine of section~\ref{sec:perfect} and our improved finite-state analysis.
We use this relationship to prove the perfect precision of our method.

\section{Background}
\label{sec:background}
Static analysis by abstract interpretation proves properties of a program by running its code through an interpreter powered by an 
\textit{abstract semantics} that approximates the behavior of a \textit{concrete semantics}.
This process is a general method for analyzing programs and serves applications such as program verification, malware/vulnerability detection, and
compiler optimization, among others 
\cite{cousot:1976:staticdetermination,cousot:1977:unifiedlatticemodel,cousot:1979:systematicdesign,midtgaard:2012:cfa}
The \textit{abstracting abstract machines} (AAM) approach uses abstract interpretation of abstract machines for \textit{control-flow analysis} (CFA) of functional (higher-order) 
programming languages \cite{johnson:2013:oaam,might:2010:aam,might:2010:free}.
The AAM methodology allows a high degree of control over how program states are represented and is easy to instrument.

In this section, we review operational semantics and abstract interpretation using AAM along with other concepts we will require as we progress.
We present a concrete interpretation of a simple direct-style language, a traditional finite-state abstraction, and a store-widened polynomial-time analysis.
We then explore the return-flow merging problem in greater detail.

\subsection{Concrete Semantics}
\label{sec:background:concrete}
We will be using the direct-style (call-by-value, untyped) $\lambda$-calculus in administrative-normal-form (ANF) \cite{flanagan:1993:essence}.
\begin{align*}
 \expr \in \syn{Exp} &\produces \letiform{\vv}{\appform{\fexpr}{\aexpr}}{\expr} && \text{[call]}
 \\
 &\;\;\opor\;\; \aexpr && \text{[return]}
 \\
 \fexpr,\aexpr \in \syn{AExp} &\produces \vv \opor \lam && \text{[atomic expressions]}
 \\
 \lam \in \syn{Lam} &\produces \lamform{\vv}{\expr} && \text{[lambda abstractions]}
 \\
 \vv,y \in \syn{Var} &\text{ is a set of identifiers} && \text{[variables]}
\end{align*}
All intermediate expressions are administratively \texttt{let}-bound, and the order of operations is made explicit as a stack of such \texttt{let}s.
This not only simplifies our semantics, but is convenient for analysis as every intermediate expression can naturally be given a unique identifier.
Additional core forms permitting mutation, recursive binding, conditional branching, tail calls, and primitive operations
add complexity, but do not complicate the technique we aim to discuss and so are left out.

Our concrete interpreter operates over machine states $\varsigma$.
\begin{align*}
 \state \in \State &\defas \syn{Exp} \times \var{Env} \times \var{Store} \times \var{Kont}  && \text{[states]}
 \\
 \env \in \var{Env} &\defas \syn{Var} \parto \var{Addr} && \text{[environments]}
 \\
 \store \in \var{Store} &\defas \var{Addr} \parto \var{Clo} && \text{[stores]}
 \\
 \clo \in \var{Clo} &\defas \syn{Lam} \times \var{Env} && \text{[closures]}
 \\
 \cont \in \var{Kont} &\defas \var{Frame}^{*} && \text{[stacks]}
 \\
 \phi \in \var{Frame} &\defas \syn{Var} \times \syn{Exp} \times \var{Env} && \text{[stack frames]}
 \\
 \addr \in \var{Addr} &\text{ is an infinite set} && \text{[addresses]}
\end{align*}
Binding environments ($\rho$) map variables in scope to a representative address ($a$).
Value stores ($\sigma$) map these addresses to a program value.
(For pure $\lambda$-calculus, all values are closures.)
Both are partial functions that are incrementally extended with new points.
A closure ($\clo$) pairs a syntactic lambda with an environment over which it is closed.
Continuations ($\cont$) are unbounded sequences of stack frames.
Each stack frame ($\phi$) contains a variable to bind,
an expression control returns to, and an environment to reinstate.
%
%
Addresses ($a$) may be drawn from any set which permits us to generate an arbitrary number of fresh values (e.g., $\mathbb{N}$).

We define a helper $\mathcal{A} : \syn{AExp} \times \var{Env} \times \var{Store} \parto \var{Clo}$ for atomic-expression evaluation: 
\begin{align*}
  \mathcal{A}(\vv, \rho, \sigma) &\defas \sigma(\rho(\vv))  && \text{[variable lookup]}
  \\
  \mathcal{A}(\lam, \rho, \sigma) &\defas (\lam, \rho)  && \text{[closure creation]}
\end{align*}

A concrete transition relation $(\cTo) : \Sigma \parto \Sigma$ defines the operation of this machine by determining at most one successor for any given
predecessor state.
The machine stops when the end of a program's execution is reached or when given an invalid state.
Call sites transition according to the following transition rule: 
\begin{align*}
 (\letiform{y}{\appform{\fexpr}{\aexpr}}{\expr}, \rho, \sigma, \cont)
 &\cTo
 (\expr', \rho', \sigma', \phi\!:\!\cont)
 \text{, where }
\end{align*}
\vspace{-0.55cm}
\begin{align*}
  \phi &= (y, \expr, \rho)
  \\
  (\lamform{x}{\expr'},\rho_\lambda) &= \mathcal{A}(\fexpr, \rho, \sigma)
  \\
  \rho' &= \rho_\lambda[x \mapsto \addr]
  \\
  \sigma' &= \sigma [\addr \mapsto \mathcal{A}(\aexpr,\rho,\sigma)]
  \\
  a & \text{ is a \textit{fresh} address}
\end{align*}
A new frame $\phi$ is pushed onto the stack for eventually returning to the body of this \texttt{let}-form.
The atomic expression $f$ is either a lambda-form or a variable-reference and is evaluated to a closure by our helper $\mathcal{A}$.
In our notation, ticks are used to uniquely name identifiers that may be different.
These do not have any bearing on the variable's domain, but where possible will hint at usage (e.g., a single tick for a successor's components).
A subscript may be more significant, but we will be careful to point it out.
This is not the case for $\rho_\lambda$, which is used to name whatever environment was drawn from the closure for $f$.
This is simply an environment distinct from $\rho$ and $\rho'$.
We generate a \textit{fresh} address $\addr$ (any address such that $\addr \notin \var{dom}(\sigma)$) and update $\rho_\lambda$ with a mapping $x \mapsto a$ to
produce the successor environment $\rho'$.
Likewise, the prior store $\sigma$ is extended at this address with the value for $\aexpr$ to produce $\sigma'$.

Return points transition according to a second rule:
\begin{align*}
 (\aexpr, \rho, \sigma, \phi\!:\!\cont)
 &\cTo
 (\expr, \rho', \sigma', \cont)\text{, where }
\end{align*}
\vspace{-0.55cm}
\begin{align*}
 \phi &= (\vv, \expr, \rho_\cont)
 \\
 \rho' &= \rho_\cont[\vv \mapsto \addr]
 \\
 \sigma' &= \sigma [ \addr \mapsto \mathcal{A}(\aexpr, \rho, \sigma) ]
 \\
 \addr & \text{ is a \textit{fresh} address}
\end{align*}
The top stack frame $\phi$ is decomposed and its environment $\rho_\cont$ extended with a fresh address $a$ to produce $\rho'$.
Likewise, the store is extended at this address with the value for $\aexpr$ to produce $\sigma'$.
The expression $\expr$ in the top stack frame is reinstated at $\rho'$, and $\sigma'$ is put atop the predecessor's stack tail $\cont$.

To fully evaluate a program $e_0$ using these transition rules, we \textit{inject} it into our state-space using a helper 
$\mathcal{I} : \syn{Exp} \to \Sigma$:
\begin{align*}
  \mathcal{I}(e) \defas (e, \varnothing, \varnothing, \epsilon)
\end{align*}
We perform the standard lifting of $(\cTo)$ to obtain a collecting semantics defined over sets of states:
\begin{align*}
  s \in S &\defas \PowSm{\Sigma}
\end{align*}
Our collecting relation $(\cToW)$ is a monotonic, total function that gives a set including the trivially reachable state $\mathcal{I}(e_0)$
plus the set of all states immediately succeeding those in its input.
\begin{align*}
  s \cToW s' &\defas s' = \set{ \varsigma'\ \vert\ \varsigma \in s \wedge \varsigma \cTo \varsigma' } \cup \set{ \mathcal{I}(e_0) }
\end{align*}

If the program $e_0$ terminates, iteration of $(\cToW)$ from $\bot$ (i.e., the empty set $\varnothing$) does as well.
That is, ${(\cToW)}^{n}(\bot)$ is a fixed point containing $e_0$'s full program trace for some $n \in \mathbb{N}$ whenever $e_0$ is a terminating program.
No such $n$ is guaranteed to exist in the general case (when $e_0$ is a non-terminating program)
as our language (the untyped $\lambda$-calculus) is Turing-complete, our semantics is fully precise, and the state-space we defined is infinite.

\subsection{Abstract Semantics}
\label{sec:background:abstract}
We are now ready to design a computable approximation of the exact program trace using an abstract semantics.
Previous work has explored a wide variety of approaches to systematically abstracting a semantics like these 
\cite{might:2010:aam,johnson:2013:oaam,might:2010:free}.
%
Broadly construed, the nature of these changes is to simultaneously finitize the domains of our machine while introducing non-determinism both into the 
transition relation (multiple successor states may immediately follow a predecessor state) and the store (multiple values may be indicated by a single address).
We use a finite state space to ensure computability.
However, to justify that a semantics defined over this finite machine is soundly approximating
our concrete semantics (for a defined notion of abstraction), we must also modify our finite states so that a potentially infinite number of 
concrete states may abstract to a single finite state.
We will use this term \textit{finite state} to differentiate from other kinds of machine states.
Components unique to this finite-state machine wear tildes:
\begin{align*}
 \tilde{\varsigma} \in \tilde{\Sigma} &\defas \syn{Exp} \times \widetilde{\var{Env}} \times \widetilde{\var{Store}} && \text{[states]}
 \\
 &\ \ \ \ \ \ \ \times \widetilde{\var{KStore}} \times \widetilde{\var{Addr}}
 \\
 \tilde{\rho} \in \widetilde{\var{Env}} &\defas \syn{Var} \parto \widetilde{\var{Addr}} && \text{[environments]}
 \\
 \tilde{\sigma} \in \widetilde{\var{Store}} &\defas \widetilde{\var{Addr}} \to \tilde{D} && \text{[stores]}
 \\
 \tilde{d} \in \tilde{D} &\defas \PowSm{\widetilde{\var{Clo}}} && \text{[flow-sets]}
 \\
 \widetilde{\clo} \in \widetilde{\var{Clo}} &\defas \syn{Lam} \times \widetilde{\var{Env}} && \text{[closures]}
 \\
 \tilde{\sigma}_{\kappa} \in \widetilde{\var{KStore}} &\defas \widetilde{\var{Addr}} \to \tilde{K} && \text{[continuation stores]}
 \\
 \tilde{k} \in \tilde{K} &\defas \PowSm{\widetilde{\var{Kont}}}  && \text{[kont-sets]}
 \\
 \tilde{\cont} \in \widetilde{\var{Kont}} &\defas \widetilde{\var{Frame}} \times \widetilde{\var{Addr}} && \text{[continuations]}
 \\
 \tilde{\phi} \in \widetilde{\var{Frame}} &\defas \syn{Var} \times \syn{Exp} \times \widetilde{\var{Env}} && \text{[stack frame]}
 \\
 \tilde{a},\tilde{a}_\cont \in \widetilde{\var{Addr}} &\text{ is a finite set} && \text{[addresses]}
\end{align*}
There were two fundamental sources of unboundedness in the concrete machine: the value store (with an infinite domain of addresses), and 
the current continuation (modeled as an unbounded list of stack frames).
We bound the value store ($\tilde{\sigma}$) by restricting its domain to a finite set of addresses ($\tilde{a}$), but we permit a 
\textit{set} of abstract closures ($\widetilde{\clo}$) at each.
We finitize the stack similarly by threading it through the store as a linked list.
A continuation is thus represented by an address.
This address points to a \textit{set} of topmost frames, each paired with the address of its continuation in turn (i.e., that stack's tail).
We separate the continuation store ($\tilde{\sigma}_\cont$) from the value store ($\tilde{\sigma}$) to maintain simplicity as we progress.

Abstract environments ($\tilde{\rho}$) change only because our address set is now finite. 
Abstract closures ($\widetilde{\clo}$) are approximate only by virtue of their environments using these abstract addresses.
For each such $\tilde{a}$, the finite value store ($\tilde{\sigma}$) denotes a \textit{flow set} ($\tilde{d}$) of closures.
At each point, a continuation store ($\tilde{\sigma}_\cont$) has a set of continuations ($\tilde{k}$).
Like closures, each abstract frame ($\tilde{\phi}$) is approximate only by virtue of its abstracted environment.
An abstract continuation ($\tilde{\cont}$) pairs a frame with an address ($\tilde{a}_\cont$) for the stack underneath.

As before, we define a helper for abstract atomic evaluation, $\tilde{\mathcal{A}}$:
\begin{align*}
  \tilde{\mathcal{A}} : \syn{AExp} \times \widetilde{\var{Env}} \times \widetilde{\var{Store}} \parto \tilde{D}
\end{align*}
\vspace{-0.5cm}
\begin{align*}
  \tilde{\mathcal{A}}(\vv, \tilde{\rho}, \tilde{\sigma}) &\defas \tilde{\sigma}(\tilde{\rho}(\vv))  && \text{[variable lookup]}
  \\
  \tilde{\mathcal{A}}(\lam, \tilde{\rho}, \tilde{\sigma}) &\defas \set{ (\lam,\tilde{\rho}) }  && \text{[closure creation]}
\end{align*}
Note that atomic evaluation of a lambda expression new yields a set containing a single element for the closure of that lambda.

Because our address domain is now finite, multiple concrete allocations need to be represented by a single abstract address.
There are a variety of sound strategies for doing this.
Each strategy corresponds to a distinct style of analysis and is amenable to easy implementation
by defining an auxiliary $\widetilde{\var{alloc}}$ helper to encapsulate these differences in behavior.
Given the variable for which to allocate and the finite state performing the allocation, the abstract allocator returns an address:
\begin{align*}
  \widetilde{\var{alloc}} : \syn{Var} \times \tilde{\Sigma} \to \widetilde{\var{Addr}}
\end{align*}
One such behavior is to simply return the variable itself (as a $0$-CFA would):
\begin{align*}
  \widetilde{\var{alloc}}_{0}(\vv,\tilde{\varsigma}) &\defas \vv
\end{align*}
Using $\widetilde{\var{alloc}}_{0}$ would tune our finite-state semantics to the \textit{monovariant} analysis style (also called zeroth-order CFA), 
a form of context-insensitive analysis.
In a monovariant analysis, every closure that is bound to a variable $x$ at any point during a concrete execution ends up being
represented in a single flow set when the analysis is complete.

Because we are also now store-allocating continuations and distinguishing a top-level continuation store, we likewise distinguish 
an abstract allocator specifically for addresses in this store:
\begin{align*}
  \widetilde{\var{alloc}}_\cont : \tilde{\Sigma} \times \syn{Exp} \times \widetilde{\var{Env}} \times \widetilde{\var{Store}} \to \widetilde{\var{Addr}}
\end{align*}
A standard choice is to allocate based on the target expression:
\begin{align*}
  \widetilde{\var{alloc}}_{\cont\ \!0}((e, \tilde{\rho}, \tilde{\sigma}, \tilde{\sigma}_\cont, \tilde{a}_\cont), e', \tilde{\rho}', \tilde{\sigma}') &\defas e'
\end{align*}
We provide to this function all the information known about the transition being made.
The value-store allocator is invoked before a successor $\tilde{\rho}'$ or $\tilde{\sigma}'$ is constructed.
However, when calling the continuation-store allocator, we provide information about the target state being transitioned to.
The choice of $e'$ for allocating a continuation address makes sense considering the entry point of a function should
know where it is returning.
In fact, when performing an analysis of a continuation-passing-style (CPS) language, $e'$ also would naturally be the choice inherited from a
monovariant value-store allocator (assuming an alpha-renaming such that every $x$ is unique to a single binding point).

We may now define a non-deterministic finite-state transition relation $(\fsTo) \subseteq \tilde{\Sigma} \times \tilde{\Sigma}$.
Call sites transition as follows.
\begin{align*}
 \overbrace{
 (\letiform{y}{\appform{\fexpr}{\aexpr}}{\expr},
  \tilde{\rho},
  \tilde{\sigma},
  \tilde{\sigma}_\cont,
  \tilde{a}_\cont)}^{\tilde{\varsigma}}
 &\fsTo
 (\expr',
  \tilde{\rho}',
  \tilde{\sigma}',
  \tilde{\sigma}_\cont',
  \tilde{a}_\cont')\text{, where }
\end{align*}
\vspace{-0.5cm}
\begin{align*}
  (\lamform{x}{\expr'},\tilde{\rho}_\lambda) &\in
   \tilde{\mathcal{A}}(\fexpr, \tilde{\rho}, \tilde{\sigma})
  \\
  \tilde{\rho}' &= \tilde{\rho}_\lambda[x \mapsto \tilde{a}]
  \\
  \tilde{\sigma}' &= \tilde{\sigma} \join [\tilde{a} \mapsto \tilde{\mathcal{A}}(\aexpr,\tilde{\rho},\tilde{\sigma})]
  \\
  \tilde{a} &= \widetilde{\var{alloc}}(x, \tilde{\varsigma})
  \\
  \tilde{\sigma}_\cont' &= \tilde{\sigma}_\cont \join [\tilde{a}'_\cont
    \mapsto \set{ ((y, \expr, \tilde{\rho}),\tilde{a}_\acont) }]
  \\
  \tilde{a}'_\cont &= \widetilde{\var{alloc}}_\cont(\tilde{\varsigma}, \expr', \tilde{\rho}', \tilde{\sigma}')
\end{align*} 
As $\tilde{\mathcal{A}}$ yields a set of abstract closures for $f$, a successor state is produced for each.
Likewise, so each point in the store accumulates all closures bound at that abstract address $\tilde{a}$ and so we faithfully over-approximate
all the addresses $a$ that $\tilde{a}$ simulates, we use a join operation when extending the store.
The join of two stores distributes point-wise as follows.
\begin{align*}
  \tilde{\sigma} \sqcup \tilde{\sigma}' &\defas \lambda \tilde{a}.\ \tilde{\sigma}(\tilde{a}) \cup \tilde{\sigma}'(\tilde{a})
  \\
  \tilde{\sigma}_\cont \sqcup \tilde{\sigma}_\cont' &\defas 
      \lambda \tilde{a}_\cont.\ \tilde{\sigma}_\cont(\tilde{a}_\cont) \cup \tilde{\sigma}_\cont'(\tilde{a}_\cont)
\end{align*}
Instead of generating a fresh address for $\tilde{a}$, we use our abstract allocation policy to select one.
To instantiate a monovariant analysis like $0$-CFA, this address is simply the syntactic variable $x$.
Likewise, we generate an address for our continuation (a new stack frame atop the current continuation) and extend the continuation store.

The return transition is modified in the same way: 
\begin{align*}
 \overbrace{(\aexpr,
  \tilde{\rho},
  \tilde{\sigma},
  \tilde{\sigma}_\cont,
  \tilde{a}_\cont)}^{\tilde{\varsigma}}
 &\fsTo
 (\expr,
  \tilde{\rho}',
  \tilde{\sigma}',
  \tilde{\sigma}_\cont,
  \tilde{a}_\cont')\text{, where }
\end{align*}
\vspace{-0.5cm}
\begin{align*}
 ((\vv,\expr,\tilde{\rho}_\cont),\tilde{a}_\cont') &\in \tilde{\sigma}_\cont(\tilde{a}_\cont)
 \\
 \tilde{\rho}' &= \tilde{\rho}_\cont[\vv \mapsto \tilde{a}]
 \\
 \tilde{\sigma}' &= \tilde{\sigma} \join [ \tilde{a} \mapsto \tilde{\mathcal{A}}(\aexpr,\tilde{\rho},\tilde{\sigma}) ]
 \\
 \tilde{a} &= \widetilde{\var{alloc}}(\vv,\tilde{\varsigma})
\end{align*}
Where multiple topmost stack frames are pointed to by $\tilde{a}_\cont$, this transition yields multiple successors.
An updated environment and store are produced as before, but the continuation store remains as it was.
The current continuation $\tilde{a}_\cont'$ reinstated in each successor is the address associated with each topmost stack frame.

To approximately evaluate a program according to these abstract semantics, we first define an abstract injection function,
$\tilde{\mathcal{I}}$, where the stores begin as functions, $\bot$, that map every abstract address to the empty set.
\begin{align*}
 \tilde{\mathcal{I}} : \syn{Exp} \to \tilde{\Sigma}
\end{align*}
\begin{align*}
 \tilde{\mathcal{I}}(\expr) \defas (\expr, \varnothing, \bot, \bot, \tilde{a}_{\text{halt}})
\end{align*}
The address $\tilde{a}_{\text{halt}}$ can be any otherwise unused address that is never returned
by the allocation function.
Our machine will eventually be unable to transition into this continuation and will then produce no successors, which simulates the
behavior of our concrete machine upon reaching an empty stack $(\epsilon)$.

We again lift $(\fsTo)$ to obtain a collecting semantics $(\fsToS)$ defined over sets of states:
\begin{align*}
  \tilde{s} \in \tilde{S} &\defas \PowSm{\tilde{\Sigma}}
\end{align*}
\vspace{-0.5cm}
\begin{align*}
  \tilde{s} \fsToS \tilde{s}'\ &\defas\ \tilde{s}' = \set{ \tilde{\varsigma}'\ \vert\ \tilde{\varsigma} \in \tilde{s} \wedge \tilde{\varsigma} \fsTo \tilde{\varsigma}' } \cup \set{ \tilde{\mathcal{I}}(e_0) }
\end{align*}
Our collecting relation $(\fsToS)$ is a monotonic, total function that gives a set including the trivially reachable finite-state $\tilde{\mathcal{I}}(e_0)$
plus the set of all states immediately succeeding those in its input.

Because $\tilde{\Sigma}$ is now finite, we know the approximate evaluation of even a non-terminating $e_0$ will terminate.
That is, for some $n \in \mathbb{N}$, the value $(\fsToS)^{n}(\bot)$ is guaranteed to be a fixed point containing an approximation of $e_0$'s full program trace
\cite{tarski:1955:fixpoint}.

\subsection{Soundness}
\label{sec:background:soundness}
An analysis is \textit{sound} if the information it provides about a program represents an accurate bound on the behavior of all possible concrete executions.
The kind of control-flow information the finite-state analysis in section~\ref{sec:background:abstract} obtains is a conservative over-approximation of 
program behavior.  It places an upper bound on the propagation of closures though a program.

To establish such a relationship between a concrete and abstract semantics, we use Galois connections.
A Galois connection is a pair of functions for abstraction and concretization such that the following holds.
\[
\begin{gathered}
\begin{aligned}
  \alpha : S \to \tilde{S}
  & & & & &
  \gamma : \tilde{S} \to S 
\end{aligned}
\\
\\
\begin{aligned}
  \alpha(s) \subseteq \tilde{s} \iff s \subseteq \gamma(\tilde{s})
\end{aligned}
\end{gathered}
\]
Using this defined notion of simulation, we may show that our abstract semantics approximates the concrete semantics by proving that 
simulation is preserved across transition:
\begin{align*}
  \alpha(s) \subseteq \tilde{s} \wedge s \cToW s' \implies \tilde{s} \fsToS \tilde{s}' \wedge \alpha(s') \subseteq \tilde{s}'
\end{align*}
Diagrammatically this is:
\begin{equation*}
  \begin{CD}
     s @> \cToW >> s'  \\
     @V \subseteq V \alpha V    @V \subseteq V \alpha V \\
    \tilde{s}  @> \fsToS >> \tilde{s}'
  \end{CD}
\end{equation*}
Both constructing analyses using Galois connections and proving them sound using Galois connections has been extensively explored in the literature 
\cite{might:2010:free,might:2010:aam}.
The analysis style we constructed in section~\ref{sec:background:abstract} has been previously proven sound using the above method \cite{might:2007:diss}.

\subsection{Store Widening}
\label{sec:background:widening}
Various forms of widening and further approximations may be layered on top of this na\"ive analysis.
One such approximation is store widening, which is necessary for our analysis to be tractable (i.e., polynomial time).
To see why store widening is necessary, let us consider the complexity of an analysis using $(\fsToS)$.
The height of the power-set lattice $(\tilde{S}, \cup, \cap)$ is the number of elements in $\tilde{\Sigma}$ which is the product of
expressions, environments, stores, and addresses.
A standard worklist algorithm at most does work proportional to the number of states it can discover \cite{might:2010:mcfa}.
For the imprecise allocators we have defined, analysis run-time is thus in:
\begin{align*}
  O( \overbrace{n}^{\vert\syn{Exp}\vert} \times \overbrace{n}^{\vert\widetilde{\var{Env}}\vert} \times \overbrace{2^{n^2}}^{\vert\widetilde{\var{Store}}\vert}
      \times \overbrace{2^{n^2}}^{\vert\widetilde{\var{KStore}}\vert} \times \overbrace{n}^{\vert\widetilde{\var{Addr}}\vert} \! )
\end{align*}
The number of syntactic points in an input program is in $O(n)$.
In the monovariant case, environments map variables to themselves and are isomorphic to the sets of free variables that may be determined for each syntactic point.
The number of addresses produced by our monovariant allocators is in $O(n)$ as these are either syntactic variables or expressions.
The number of value stores may be visualized as a table of possible mappings from every address to every abstract closure---each may be 
included in a given store or not as seen in figure~\ref{figure:value-space-of-stores}.
\def\matriximgbigger{%
  \kbordermatrix{& \widetilde{clo}_0 & \widetilde{clo}_1 & \widetilde{clo}_2 & \cdots & \widetilde{clo}_i & \cdots \cr
	               \tilde{a}_0 & 0 & 1 & 0 & \cdots & 0 & \cdots \cr
                       \tilde{a}_1 & 1 & 0 & 0 & \cdots & 1 & \cdots \cr
                       \tilde{a}_2 & 0 & 0 & 1 & \cdots & 0 & \cdots \cr
		       \vdots & \vdots & \vdots & \vdots & \ddots & \vdots & \vdots \cr
                       \tilde{a}_j & 1 & 0 & 0 & \cdots & 1 & \cdots \cr
                       \vdots & \vdots & \vdots & \vdots & \vdots & \vdots & \ddots \cr }
}
\def\matriximg{%
  \kbordermatrix{& \widetilde{clo}_0 & \widetilde{clo}_1 & \cdots & \widetilde{clo}_i & \cdots \cr
	               \tilde{a}_0 & 0 & 0 & \cdots & 1 & \cdots \cr
                       \tilde{a}_1 & 1 & 1 & \cdots & 0 & \cdots \cr
		       \vdots & \vdots & \vdots & \ddots & \vdots & \vdots \cr
                       \tilde{a}_j & 0 & 0 & \cdots & 1 & \cdots \cr
                       \vdots & \vdots & \vdots & \vdots & \vdots & \ddots \cr }
}
\begin{figure}
\[
  \text{\scriptsize $O(n)$}\left\{\left.\vphantom{\raisebox{-0.9\baselineskip}{\scalebox{0.8}[0.8]{\matriximg}}}\right.\right.\kern-2\nulldelimiterspace
\hspace{0.71cm}
\overbrace{\phantom{\scalebox{0.61}[1.0]{\matriximg}}}^{O(n)}\kern-\nulldelimiterspace\vphantom{\matriximg}
\hspace{-4.7cm}\matriximg
\]
\caption{The value space of stores.}
\label{figure:value-space-of-stores}
\end{figure}
The number of abstract closures is in $O(n)$ because lambdas uniquely determine a monovariant environment.
This times the number of addresses gives $O(n^2)$ possible additions to the value store.
The number of continuations is likewise in $O(n)$ because $\texttt{let}$-forms uniquely determine their binding variable, 
body, and monovariant environment.
This times the number of possible addresses gives $O(n^2)$ possible additions to the continuation store.

The crux of the issue is that, in exploring a na\"ive state-space (where each state is specific to a whole store), we may explore 
both sides of every diamond in the store lattices.
All combinations of possible bindings in a store may need to be explored, including every alternate path up the store lattice.
For example, along one explored path we might extend an address $\tilde{a}_1$ with $\widetilde{\clo}_1$ before extending it with $\widetilde{\clo}_2$, 
and along another path we might add these closures in the reverse order (i.e., $\widetilde{\clo}_2$ before $\widetilde{\clo}_1$). 
We might also extend another address $\tilde{a}_2$ with $\widetilde{\clo}_1$ either before or after either of these cases, and so forth.
This potential for exponential blow-up is unavoidable without further widening or coarser structural abstraction.

Global-store widening is an essential technique for combating exponential blow up.
This lifts the store alongside a set of reachable states instead of nesting them inside states in $\tilde{\Sigma}$.
To formalize this, we define new \textit{widened} state spaces that pair a set of reachable \textit{configurations} (states \textit{sans} stores) 
with a global value store and global continuation store.
Instead of accumulating whole stores, and thereby all possible sequences of additions within such stores,
the analysis strictly accumulates new values in the store in the same way $(\fsToS)$ accumulates reachable states in an $\tilde{s}$:
\begin{align*}
 \tilde{\xi} \in \tilde{\Xi} &\defas \tilde{R} \times \widetilde{\var{Store}} \times \widetilde{\var{KStore}} && \text{[state-spaces]}
 \\
 \tilde{r} \in \tilde{R} &\defas \PowSm{ \tilde{C} }  && \text{[reachable configurations]}
 \\
 \tilde{c} \in \tilde{C} &\defas \syn{Exp} \times \widetilde{\var{Env}} \times \widetilde{\var{Addr}}  && \text{[configurations]}
\end{align*}
A widened transfer function $(\fsToW)$ may then be defined that, like $(\fsToS)$, is a monotonic, total function we may iterate to a fixed point.
\begin{align*}
  (\fsToW) : \tilde{\Xi} \to \tilde{\Xi}
\end{align*}
This may be defined in terms of $(\fsTo)$, as was $(\fsToS)$, by transitioning each reachable configuration using the global store to yield a new set of 
reachable configurations and a set of stores whose least upper bound is the new global store: 
\begin{align*}
 (\tilde{r},
  \tilde{\sigma},
  \tilde{\sigma}_\cont)
 &\fsToW
 (\tilde{r}',
  \tilde{\sigma}',
  \tilde{\sigma}_\cont')\text{, where }
\end{align*}
\vspace{-0.5cm}
\begin{align*}
  \tilde{s} &= \set{ \tilde{\varsigma}\ \vert\ (e, \tilde{\rho}, \tilde{a}_\cont) \in \tilde{r} \wedge (e, \tilde{\rho}, \tilde{\sigma}, \tilde{\sigma}_\cont, \tilde{a}_\cont) \fsTo \tilde{\varsigma} } \cup \set{ \tilde{\mathcal{I}}(e_0) }
  \\
  \tilde{r}' &= \set{ (e, \tilde{\rho}, \tilde{a}_\cont)\ \vert\ (e, \tilde{\rho}, \tilde{\sigma}, \tilde{\sigma}_\cont, \tilde{a}_\cont) \in \tilde{s} }
  \\
  \tilde{\sigma}' &= \!\!\!\!\!\!\!\!\bigsqcup_{(\_, \_, \tilde{\sigma}'', \_, \_) \in \tilde{s}}\!\!\!\!\!\!\!\tilde{\sigma}''
  \ \ \ \ \ \ \ \ \ \ \ \ \ \ \ \ 
  \tilde{\sigma}_\cont' = \!\!\!\!\!\!\!\!\!\bigsqcup_{(\_, \_, \_, \tilde{\sigma}_\cont'', \_) \in \tilde{s}}\!\!\!\!\!\!\!\!\tilde{\sigma}_\cont''
\end{align*}
In this definition, an underscore (wildcard) matches anything.
The height of the $\tilde{R}$ lattice is linear (as environments are monovariant) and the height of the store lattices are quadratic (as each global store is strictly extended).
Each extension of the store may require $O(n)$ transitions because at any given store, we must transition every configuration to be sure to obtain any changes to the store or
otherwise reach a fixed point.
A traditional worklist algorithm for computing a fixed point is thus cubic: 
\begin{align*}
  O( \overbrace{n}^{\vert\tilde{C}\vert}\ \times\ (\!\!\overbrace{n^2}^{\vert\widetilde{\var{Store}}\vert} + \overbrace{n^2}^{\vert\widetilde{\var{KStore}}\vert} \!\!\!\! ))
\end{align*}

\subsection{Stack Imprecision}
\label{sec:background:imprecision}
To illustrate the effect of an imprecise stack on data-flow and control-flow precision, we first define a more precise $1$-call-sensitive 
(first-order, $1$-CFA) allocator.
A $k$-call-sensitive analysis style differentiates bindings to a variable so they are unique to a history of the last 
$k$ call sites reached before the binding.
A history of length $k=1$ then allocates an address unique to the call site immediately preceeding the binding by using the following allocator.
\begin{align*}
  \widetilde{alloc}_1(x, (e, \tilde{\rho}, \tilde{\sigma}, \tilde{\sigma}_\cont, \tilde{\cont})) &\defas (x, e)
\end{align*}
Now, using $\widetilde{\var{alloc}}_1$, consider the following snippet of code where the variable \texttt{id} is already bound to 
\texttt{($\lambda$ (x) \!\textsuperscript{0}x)}:
\begin{Verbatim}[commandchars=\\\{\}]
... \textsuperscript{1}(let ([y (id #t)])
      \textsuperscript{2}(let ([z (id #f)])
        \textsuperscript{3}...))
\end{Verbatim}
We number these expressions for ease of reference.
For example, $e_2$ refers to the \texttt{let}-form that binds \texttt{z}, and $e_0$ to the return point of \texttt{id}. 
We assume the starting configuration for this example is $(e_1, \tilde{\rho}, \tilde{a}_\cont)$ where
$\tilde{\rho}$ and $\tilde{\addr}_\cont$ are the binding environment and continuation address at the start of this code.
We likewise let $\tilde{\rho}_\lambda$ be the environment of \texttt{id}'s closure.

The first call to \texttt{id} transitions to evaluate $e_0$ with the continuation address $e_0$. 
This transition reaches the configuration $(e_0, \tilde{\rho}_\lambda[\texttt{x} \mapsto (\texttt{x}, e_1)], e_0)$ and binds $(\texttt{x}, e_1)$ 
to \texttt{\#t} and the continuation address $e_0$ to the continuation $((\texttt{y},e_{2},\tilde{\rho}),\tilde{\addr}_{\kappa})$, 
which gives us the following stores:
\begin{alignat*}{5}
 & \tilde{\sigma} &  & =\{ &  & \!\left(\texttt{x},e_{1}\right) &  & \mapsto &  & \begin{alignedat}[t]{1}\{ & \texttt{\#t}\}\end{alignedat}
\}\\
 & \tilde{\sigma}_{\kappa} &  & =\{ &  & e_{0} &  & \mapsto &  & \begin{alignedat}[t]{1}\{ & \!\left((\texttt{y},e_{2},\tilde{\rho}),\tilde{\addr}_{\kappa}\right)\}\}\end{alignedat}
\end{alignat*}
Next, \texttt{id} returns and transitions from $e_0$ to $e_2$, extending the continuation's environment to $\tilde{\rho}[\texttt{y}\mapsto(\texttt{y}, e_0)]$
and reinstating the continuation address $\tilde{a}_\cont$.
This yields a configuration $(e_2, \tilde{\rho}[\texttt{y}\mapsto(\texttt{y}, e_0)], \tilde{a}_\cont)$.
This transition binds $(\texttt{y}, e_0)$ to \texttt{\#t}, giving us the following stores:
\begin{alignat*}{5}
 & \tilde{\sigma} &  & =\{ &  & \!\left(\texttt{x},e_{1}\right) &  & \mapsto &  & \begin{alignedat}[t]{1}\{ & \texttt{\#t}\}\end{alignedat}
,\\
 &  &  &  &  & \!\left(\texttt{y},e_{0}\right) &  & \mapsto &  & \begin{alignedat}[t]{1}\{ & \texttt{\#t}\}\end{alignedat}
\}\\
 & \tilde{\sigma}_{\kappa} &  & =\{ &  & e_{0} &  & \mapsto &  & \begin{alignedat}[t]{1}\{ & \!\left((\texttt{y},e_{2},\tilde{\rho}),\tilde{\addr}_{\kappa}\right)\}\}\end{alignedat}
\end{alignat*}
Then the second call to \texttt{id} transitions to evaluate $e_0$ with the continuation address $e_0$ once again 
(recall the definition of $\widetilde{\var{alloc}}_{\cont\ \!0}$).
This transition reaches the configuration $(e_0, \tilde{\rho}_\lambda[\texttt{x} \mapsto (\texttt{x}, e_2)], e_0)$, binding $(\texttt{x}, e_2)$ to \texttt{\#f} and the continuation address $e_0$ to the continuation $((\texttt{z},e_{3},\tilde{\rho}[\texttt{y}\mapsto(\texttt{y},e_0)]),\tilde{\addr}_{\kappa})$, giving us the following stores:
\begin{alignat*}{5}
 & \tilde{\sigma} &  & =\{ &  & \!\left(\texttt{x},e_{1}\right) &  & \mapsto &  & \begin{alignedat}[t]{1}\{ & \texttt{\#t}\}\end{alignedat}
,\\
 &  &  &  &  & \!\left(\texttt{y},e_{0}\right) &  & \mapsto &  & \begin{alignedat}[t]{1}\{ & \texttt{\#t}\}\end{alignedat}
,\\
 &  &  &  &  & \!\left(\texttt{x},e_{2}\right) &  & \mapsto &  & \begin{alignedat}[t]{1}\{ & \texttt{\#f}\}\end{alignedat}
\}\\
 & \tilde{\sigma}_{\kappa} &  & =\{ &  & e_{0} &  & \mapsto &  & \begin{alignedat}[t]{1}\{ & \!\left((\texttt{y},e_{2},\tilde{\rho}),\tilde{\addr}_{\kappa}\right),\\
 & \!\left((\texttt{z},e_{3},\tilde{\rho}\left[\texttt{y}\mapsto(\texttt{y},e_0)\right]),\tilde{\addr}_{\kappa}\right)\}\}
\end{alignedat}
\end{alignat*}
Next, \texttt{id} returns and transitions from $e_0$ to $e_3$, reinstating the continuation address $\tilde{a}_\cont$ and extending the continuation's
environment to $\tilde{\rho}[\texttt{y}\mapsto(\texttt{y},e_0)][\texttt{z}\mapsto(\texttt{z},e_0)]$.
Because $e_0$ is bound to two continuations, this transition binds $(\texttt{z}, e_0)$ to \texttt{\#f} while another spuriously binds
$(\texttt{y}, e_0)$ to \texttt{\#f}, causing return-flow imprecision in the following stores:
\begin{alignat*}{5}
 \\
 & \tilde{\sigma} &  & =\{ &  & \!\left(\texttt{x},e_{1}\right) &  & \mapsto &  & \begin{alignedat}[t]{1}\{ & \texttt{\#t}\}\end{alignedat}
,\\
 &  &  &  &  & \!\left(\texttt{x},e_{2}\right) &  & \mapsto &  & \begin{alignedat}[t]{1}\{ & \texttt{\#f}\}\end{alignedat}
,\\
 &  &  &  &  & \!\left(\texttt{y},e_{0}\right) &  & \mapsto &  & \begin{alignedat}[t]{1}\{ & \texttt{\#t},\texttt{\#f}\}\end{alignedat}
,\\
 &  &  &  &  & \!\left(\texttt{z},e_{0}\right) &  & \mapsto &  & \begin{alignedat}[t]{1}\{ & \texttt{\#f}\}\end{alignedat}
\}\\
 \\
 & \tilde{\sigma}_{\kappa} &  & =\{ &  & e_{0} &  & \mapsto &  & \begin{alignedat}[t]{1}\{ & \!\left((\texttt{y},e_{2},\tilde{\rho}),\tilde{\addr}_{\kappa}\right),\\
 & \!\left((\texttt{z},e_{3},\tilde{\rho}\left[\texttt{y}\mapsto(\texttt{y},e_0)\right]),\tilde{\addr}_{\kappa}\right)\}\}
\end{alignedat}
 \\
\end{alignat*}
The address $(\texttt{y}, e_0)$, representing \texttt{y} within $e_3$, maps to both \texttt{\#t} and \texttt{\#f}, even though no 
concrete execution binds \texttt{y} to \texttt{\#f}.
A similar pair of transitions from $(e_0, \tilde{\rho}_\lambda[\texttt{x} \mapsto (\texttt{x}, e_1)], e_0)$ (the second of which is prompted by a change in 
the global continuation store at the address $e_0$) cause the same conflation for \texttt{z}.

Clearly, one solution is to increase the context sensitivity of our continuation allocator.
Consider a continuation allocator $\widetilde{\var{alloc}}_{\cont\ \!1}$ that like $\widetilde{\var{alloc}}_1$ uses a single call site of context 
and allocates a continuation address $(e', e)$ formed from both the expression being transitioned to, $e'$, and 
the expresson being transitioned from, $e$. 
This results in no spurious merging at return points because continuations are kept as distinct as the 
$1$-call-sensitive value-store addresses we allocate.

It seems reasonable from here to suspect that perfect stack precision could always be obtained through a sufficiently precise 
strategy for polyvariant continuation allocation.
The difficulty is in knowing how to obtain this in the general case given an arbitrary value-store allocation strategy. 
Given that CFA2 and PDCFA promise a fixed method for implementing perfect stack precision, albeit at significant engineering and run-time costs, 
can perfect stack precision be implemented as a \textit{fixed}, adaptive continuation allocator?
In this paper, we both answer this question in the affirmative and show that this leads us not only to a trivial implementation
but to only a constant-factor increase in run-time complexity.

\vspace{2em}

\section{Perfect Stack Precision}
\label{sec:perfect}
We next formalize what is meant by a static analysis with \textit{perfect stack precision} by using an abstract abstract machine (AAM) 
\cite{might:2010:aam} with unbounded stacks within each machine configuration. 
We then review the existing polynomial-time methods for computing an analysis with equivalent precision to this machine: PDCFA and AAC.

\subsection{Unbounded-Stack Analysis}
In the same manner as previous work on this topic,
we formalize perfect stack precision
using a static analysis that
leaves the structure of stacks fully unabstracted.
Each frame of this unbounded stack is itself abstract because its environment is abstract and references the abstracted value store.
States and configurations, however, directly contain lists of such frames that are unbounded in length. 
Environments, closures, stack frames, flow sets, and value stores are otherwise abstracted in the same manner as the finite machine 
of section~\ref{sec:background:abstract}.
To differentiate this from the machines we have seen so far, we call this an unbounded-stack machine.
Components unique to this machine wear hats:

\begin{align*}
 \astate \in \aState &\defas \syn{Exp} \times \sa{Env} \times \sa{Store} \times \sa{Kont}  && \text{[states]}
 \\
 \aenv \in \sa{Env} &\defas \syn{Var} \parto \sa{Addr} && \text{[environments]}
 \\
 \astore \in \sa{Store} &\defas \sa{Addr} \to \hat{D} && \text{[stores]}
 \\
 \hat{d} \in \hat{D} &\defas \PowSm{ \widehat{Clo} }  && \text{[flow-sets]}
 \\
 \aclo \in \sa{Clo} &\defas \syn{Lam} \times \sa{Env} && \text{[closures]}
 \\
 \acont \in \sa{Kont} &\defas \sa{Frame}^{*} && \text{[whole stacks]}
 \\
 \hat{\phi} \in \sa{Frame} &\defas \syn{Var} \times \syn{Exp} \times \sa{Env} && \text{[stack frames]}
 \\
 \aaddr \in \sa{Addr} &\text{ is a finite set} && \text{[addresses]}
\end{align*}
Our atomic-expression evaluator works just as before:
\begin{align*}
  \hat{\mathcal{A}} : \syn{AExp} \times \sa{Env} \times \sa{Store} \parto \hat{D}
\end{align*}
\vspace{-0.5cm}
\begin{align*}
  \hat{\mathcal{A}}(\vv, \hat{\rho}, \hat{\sigma}) &\defas \hat{\sigma}(\hat{\rho}(\vv))  && \text{[variable lookup]}
  \\
  \hat{\mathcal{A}}(\lam, \hat{\rho}, \hat{\sigma}) &\defas \set{ (\lam, \hat{\rho}) }  && \text{[closure creation]}
\end{align*}
As does a monovariant allocator:
\begin{align*}
  \sa{alloc} : \syn{Var} \times \hat{\Sigma} \to \sa{Addr}
\end{align*}
\vspace{-0.5cm}
\begin{align*}
  \sa{alloc}_{0}(x,\hat{\varsigma}) &\defas x
\end{align*}
This may be tuned to any other allocation strategy as easily as before.

We now define a non-deterministic unbounded-stack-machine transition relation $(\pdTo) \subseteq \hat{\Sigma} \times \hat{\Sigma}$
and a rule for call-site transitions:
\begin{align*}
 \overbrace{
 (\letiform{y}{\appform{\fexpr}{\aexpr}}{\expr},
  \aenv,
  \astore,
  \acont)}^{\astate}
 &\pdTo
 (\expr',
  \aenv',
  \astore',
  \hat{\phi}\!:\!\acont)\text{, where }
\end{align*}
\vspace{-0.5cm}
\begin{align*}
 \hat{\phi} &= (y, \expr, \aenv)
 \\
 (\lamform{x}{\expr'},\aenv_\lambda) &\in
  \aArgEval(\fexpr, \aenv, \astore)
 \\
 \aenv' &= \aenv_\lambda [ x \mapsto \aaddr ]
 \\
 \astore' &= \astore \join [\aaddr \mapsto \aArgEval(\aexpr,\aenv,\astore)]
 \\
 \aaddr &= \aalloc(x,\astate)
\end{align*}
This is slightly simplified from its analogue in $(\fsTo)$.
The definitions of $e'$, $\hat{\rho}'$, and $\hat{\sigma}'$ are effectively identical, but the continuation store and continuation address have been replaced with an unbounded stack $\hat{\phi}\!:\!\hat{\cont}$.

Likewise, the return transition also changes to the following.
\begin{align*}
 \overbrace{(\aexpr,
  \aenv,
  \astore,
  \hat{\phi}\!:\!\acont)}^{\astate}
 &\pdTo
 (\expr,
  \aenv',
  \astore',
  \acont)\text{, where }
\end{align*}
\vspace{-0.5cm}
\begin{align*}
 \hat{\phi} &= (\vv, \expr, \aenv_\cont)
 \\
 \aenv' &= \aenv_\cont[\vv \mapsto \aaddr]
 \\
 \astore' &= \astore \join [\aaddr \mapsto \aArgEval(\aexpr,\aenv,\astore)]
 \\
 \aaddr &= \aalloc(\vv,\astate)
\end{align*}
To follow a return transition, the stack must contain at least one frame. 
Then the appropriate $\expr$ is reinstated with the environment $\hat{\rho}$ extended with an address for $\vv$. 
The store is extended and whatever stack tail existed after $\hat{\phi}$ is the successor's continuation $\hat{\cont}$.

Unbounded-state injection is defined as we would expect:
\begin{align*}
 \hat{\mathcal{I}} : \syn{Exp} \to \hat{\Sigma}
\end{align*}
\vspace{-0.5cm}
\begin{align*}
 \hat{\mathcal{I}}(\expr) \defas (\expr, \varnothing, \bot, \epsilon)
\end{align*}
As before, we lift $(\pdTo)$ to obtain a monotonic na\"ive collecting relation $(\pdToS)$ for a program $e_0$ that is 
defined over sets of unbounded-states:
\begin{align*}
  \hat{s} \in \hat{S} &\defas \PowSm{\hat{\Sigma}}
\end{align*}
\vspace{-0.5cm}
\begin{align*}
  \hat{s} \pdToS \hat{s}' &\defas \hat{s}' = \set{ \hat{\varsigma}'\ \vert\ \hat{\varsigma} \in \hat{s} \wedge \hat{\varsigma} \pdTo \hat{\varsigma}' } \cup \set{ \hat{\mathcal{I}}(e_0) }
\end{align*}
This analysis is approximate but remains incomputable because the stack can grow without bound.
Put another way, the height of the lattice $(\hat{S}, \cup, \cap)$ is infinite and so no finite number of $(\pdToS)$-iterations is guaranteed to obtain a fixed point.

\subsection{Store-Widened Unbounded-Stack Analysis}
As we will be comparing this unbounded-stack analysis to our new technique using precise store-allocated continuations, 
we derive a global-store-widened version as before:
\begin{align*}
 \hat{\xi} \in \hat{\Xi} &\defas \hat{R} \times \sa{Store}  && \text{[state-spaces]}
 \\
 \hat{r} \in \hat{R} &\defas \PowSm{\sa{C}}  && \text{[reachable configs.]}
 \\
 \hat{c} \in \hat{C} &\defas \syn{Exp} \times \sa{Env} \times \sa{Kont}  && \text{[configurations]}
\end{align*}
A widened transfer function $(\pdToW)$ is defined in terms of \mbox{$(\pdTo)$} in exactly the same manner as $(\fsToW)$ was derived from $(\fsTo)$ except that we now
have only a single global value store and no continuation store:
\begin{align*}
  (\pdToW) : \hat{\Xi} \to \hat{\Xi}
\end{align*}
\vspace{-0.5cm}
\begin{align*}
 (\hat{r},
  \hat{\sigma})
 &\pdToW
 (\hat{r}',
  \hat{\sigma}')\text{, where }
\end{align*}
\vspace{-0.5cm}
\begin{align*}
  \hat{s} &= \set{ \hat{\varsigma}\ \vert\ (e, \hat{\rho}, \hat{\cont}) \in \hat{r} \wedge (e, \hat{\rho}, \hat{\sigma}, \hat{\cont}) \pdTo \hat{\varsigma} } 
	\cup \set{ \hat{\mathcal{I}}(e_0) }
  \\
  \hat{r}' &= \set{ (e, \hat{\rho}, \hat{\cont})\ \vert\ (e, \hat{\rho}, \hat{\sigma}, \hat{\cont}) \in \hat{s} }
  \\
  \hat{\sigma}' &= \!\!\!\!\!\!\bigsqcup_{(\_, \_, \hat{\sigma}'', \_) \in \hat{s}}\!\!\!\!\!\hat{\sigma}''
\end{align*}

\subsection{Pushdown Control-Flow Analysis (PDCFA)}
Pushdown control-flow analysis (PDCFA) is a strategy for creating a computable equivalent to the precision of our unbounded-stack machine
at a quadratic-factor increase to the complexity class of the underlying finite analysis (e.g., monovariant or $1$-call-sensitive) 
\cite{earl:2010:pushdown}.
This strategy tracks both reachable states (or in the store-widened case, configurations) as well as push or pop edges between them.
A quadratic blow up comes from the fact that each pair of reachable states may have an explicitly-tracked edge between them.
These edges implicitly represent, as possible paths through the graph, the stacks explicitly represented in the unbounded-stack machine.
This graph precisely describes the regular expression of all stacks reachable in the pushdown states of the unbounded-stack analysis.

PDCFA formalizes a \textit{Dyke state graph} for this.
Where a sequence of pushes may be repeated \textit{ad infinitum}, a Dyke state graph explicitly represents
a cycle of push edges and a cycle of pop edges finitely. 
Broadly speaking, this is also how AAC and our adaptive continuation allocator work, except that such cycles are represented
in the store instead of the state graph.
A Dyke state graph is a state transition graph where each edge is annotated with either a frame push, a frame pop, or an epsilon.
The set of continuations for a particular state in a Dyke state graph is determined by the pushes and pops along the paths 
that reach that state.

To formalize these Dyke state graphs, we reuse some components of our unbounded-stack machine, continuing to use hats 
as these machines are closely related:
\begin{align*}
  \hat{g} \in \hat{G} &\defas \hat{V} \times \hat{E} \times \sa{Store}  && \text{[Dyke graph]} 
  \\
  \hat{v} \in \hat{V} &\defas \PowSm{ Q }  && \text{[Dyke vertices]}
  \\
  \hat{q} \in \hat{Q} &\defas \syn{Exp} \times \sa{Env}  && \text{[Dyke configs.]} 
  \\
  \hat{e} \in \hat{E} &\defas \PowSm{ \hat{Q} \times \sa{Frame}_{\pm} \times \hat{Q} }  && \text{[Dyke edges]} 
  \\
  \hat{\phi}_\pm \in \sa{Frame}_\pm &\defas \sa{Frame} \times \set{ \textbf{push}, \textbf{pop} } && \text{[edge actions]} 
\end{align*}
For readability, we style an edge $(\hat{q}, (\hat{\phi}, \textbf{push}), \hat{q}') \in \hat{e}$ like so:
\begin{align*}
  \hat{q} \overset{\hat{\phi}^{+}}{\longrightarrow} \hat{q}' \in \hat{e}
\end{align*}

It would be too verbose to formalize all the machinery required to compute a valid Dyke state graph.
Instead, we define it from a completed unbounded-stack analysis $\hat{\xi}$.
The function $\mathcal{DSG} : \hat{\Xi} \to \hat{G}$ produces a Dyke state graph from a fixed-point $\hat{\xi}$ for $(\pdToW)$.
The graph $\hat{g} = \mathcal{DSG}(\hat{\xi})$ is a valid Dyke state graph analysis for a program $e_0$ 
when $\hat{\xi}$ is the unbounded-stack analysis of $e_0$.
\begin{align*}
  \mathcal{DSG}(\overbrace{(\hat{r}, \hat{\sigma})}^{\hat{\xi}}) &\defas \overbrace{(\hat{v}, \hat{e}, \hat{\sigma})}^{\hat{g}}, \text{where}
\end{align*}
\vspace{-0.5cm}
\begin{align*}
  \hat{v} &= \set{ (\expr, \hat{\rho})\ \ \vert\ \ (\expr, \hat{\rho}, \hat{\cont}) \in \hat{r} }
  \\
  \hat{e} &= \set{ (\expr, \hat{\rho}) \overset{\hat{\phi}^{+}}{\longrightarrow} (\expr', \hat{\rho}')\ \ \vert\ \ (\expr, \hat{\rho}, \hat{\cont}) \in \hat{r} 
  \\
	&\ \ \ \ \ \ \ \ \ \ \ \ \ \ \ \ \wedge (\expr, \hat{\rho}, \hat{\sigma}, \hat{\cont}) \pdTo (\expr', \hat{\rho}', \hat{\sigma}, \hat{\phi}\!:\!\hat{\cont}) }
  \\
  &\ \cup \set{ (\expr, \hat{\rho}) \overset{\hat{\phi}^{-}}{\longrightarrow} (\expr', \hat{\rho}')\ \ \vert\ \ (\expr, \hat{\rho}, \hat{\cont}) \in \hat{r} 
  \\
	&\ \ \ \ \ \ \ \ \ \ \ \ \ \ \ \ \wedge (\expr, \hat{\rho}, \hat{\sigma}, \hat{\phi}\!:\!\hat{\cont}) \pdTo (\expr', \hat{\rho}', \hat{\sigma}, \hat{\cont}) } 
\end{align*}

Although we do not formalize transition relations for Dyke state graphs themselves, it will be helpful for us to illustrate the major source of
additional complexity in engineering a PDCFA directly. 
In the finite-state analysis, a transition is able to trivially compute a set of stacks by looking up the current 
continuation address in the continuation store.
In the unbounded-stack analysis, a transition is able to trivially compute the stack by looking at the final component of
the state or configuration being transitioned.
In a Dyke state graph, canceling sequences of pushes and pops may place the set of topmost stack frames on edges arbitrarily distant from the configuration $\hat{q}$
being transitioned.
In this way, the implicitness of stacks in a Dyke state graph obfuscates one of the most common operations needed to 
compute the analysis (i.e., stack introspection).
As an example, observe how the topmost stack frame $\hat{\phi}_0$ for $\hat{q}_3$ is located elsewhere in the graph:
\begin{equation*}
  \xymatrix{
    \hat{q}_0 \ar[r]^{\hat{\phi}_0^{+}} &
    \hat{q}_1 \ar[r]^{\hat{\phi}_1^{+}} &
    \hat{q}_2 \ar[r]^{\hat{\phi}_1^{-}} &
    \hat{q}_3 
  } 
\end{equation*}

PDCFA therefore requires a non-trivial algorithm for stack introspection \cite{earl:2012:introspectivepushdown} and extra analysis 
machinery overall.
Specifically, PDCFA requires the inductive maintenance of an \textit{epsilon closure graph} in addition to the Dyke state graph
as seen in the following.
\begin{equation*}
  \xymatrix{
    \hat{q}_0 \ar[r]^{\hat{\phi}_0^{+}} &
    \hat{q}_1 \ar[r]^{\hat{\phi}_1^{+}} \ar @{-->} @(u,u) [rr]^\epsilon &
    \hat{q}_2 \ar[r]^{\hat{\phi}_1^{-}} &
    \hat{q}_3 
  } 
\end{equation*}
This structure makes all sequences of canceling stack actions explicit as an epsilon edge.
As we will see, this epsilon closure graph represents unnecessary additional complexity for both computer and analysis developer.

\subsection{Precise Allocation of Continuations (AAC)}
Abstracting abstract control (AAC)~\cite{johnson:2014:aac} is another polynomial-time method for obtaining perfect stack precision.
This technique works by store-allocating continuations using addresses unique enough to ensure no spurious merging and, like PDCFA, does not require
foreknowledge of the polyvariance (e.g., context sensitivity) being used in the value store.
The method is worse than PDCFA's quadratic-factor increase in run-time complexity.
In the monovariant and store-widened case, its authors believe it to be in $O(n^8)$~\cite{johnson:2015:complexityemail}. 
However, AAC makes perfect stack precision available \textit{for free} in terms of development cost (i.e., labor).

Given the standard finite-state abstraction we built up in section~\ref{sec:background:abstract}, we can define AAC's essential strategy in a single line:
\begin{align*}
  \widetilde{\var{alloc}}_{\cont\ \!\scalebox{0.5}{\text{AAC}}}((\expr, \tilde{\rho}, \tilde{\sigma}, \tilde{\cont}), \expr', \tilde{\rho}', \tilde{\sigma}') 
	&\defas (\expr', \tilde{\rho}', \expr, \tilde{\rho}, \tilde{\sigma})
\end{align*}
That is, continuations are stored at an address unique to the target state's expression $e'$ and environment $\tilde{\rho}'$ as well as the source state's 
expression $e$, environment $\tilde{\rho}$, and store $\tilde{\sigma}$.

We have simplified AAC slightly and translated its notation to give this definition in the terms of our framework.
A more faithful presentation of AAC shows fundamental differences between their framework and ours.
AAC uses an eval-apply semantics and explodes each flow set into a set of distinct states across every application.
The exact address AAC proposes using is $((\lamform{y}{e'}, \tilde{\rho}_\lambda), \widetilde{clo}, \tilde{\sigma})$ (Figure~7 in \cite{johnson:2014:aac}) 
where $(\lamform{y}{e'}, \tilde{\rho}_\lambda)$ is the target closure of an application, $\widetilde{\clo}$ is one particular abstract closure flowing to $y$, 
and $\tilde{\sigma}$ is the value store in the source state. 
Our components $e'$ and $\tilde{\rho}'$ are isomorphic to the target closure in the sense that $e'$ is identical and $\tilde{\rho}'$ is produced from 
the combination of $\tilde{\rho}_\lambda$ and $y$.
The source state's components ($e$, $\tilde{\rho}$, $\tilde{\sigma}$) are not as specific as $\widetilde{\clo}$ and $\tilde{\sigma}$, but
they do uniquely determine a flow set $\tilde{d}$ (the result of $\tilde{\mathcal{A}}$ invoked on $f$) that contains $\widetilde{\clo}$.
However, a semantics using an eval-apply factoring like AAC is needed
to obtain a unique continuation address for every closure propagated across an application.
This would have significantly complicated our presentation of the finite-state analysis, and in section~\ref{sec:forfree} we will see that being specific to $\widetilde{\clo}$
adds run-time complexity to an analysis without adding any precision.

The intuition for AAC is that by allocating continuations specific to both the source state and target state of a call-site transition, no merging may occur when returning
according to this (transition-specific) continuation-address.
If we were to add some arbitrary additional context sensitivity (e.g., $3$-call-sensitivity), this information would be encoded in $\tilde{\rho}'$
and inherited by $\widetilde{\var{alloc}}_{\cont\ \!\scalebox{0.5}{\text{AAC}}}$ upon producing an address.
Including this target-state binding environment in continuation addresses is the key reason why AAC allocates precise continuation addresses.

In section~\ref{sec:forfree}, we will see that only the target state's expression $e'$ and environment $\tilde{\rho}'$ are truly necessary 
for obtaining the perfect stack precision of our unbounded-stack machine.
Including components of the transition's source state, its store, or its flow set only adds run-time complexity that is unnecessary for achieving perfect stack precision.
This optimization extends AAC's core insight to be computationally \textit{for free} while remaining precise and developmentally \textit{for free}.

\section{Perfect Stack Precision for Free}
\label{sec:forfree}
The primary intuition of our work can be illustrated by considering a set of intraprocedural configurations for some 
function invocation as in the following with $\tilde{c}_0$ through $\tilde{c}_5$.
\begin{center}
\begin{tikzpicture}[node distance=0.4cm]
\path node (caller) {};
\path node (c0) [below right=of caller] {$\tilde{c}_0$};
\path node (c1) [right=of c0] {$\tilde{c}_1$};
\path node (g0) [below right=of c1] {$f$};
\path node (c2) [above right=of g0] {$\tilde{c}_2$};
\path node (c3) [right=of c2] {$\tilde{c}_3$};
\path node (h0) [below right=of c3] {$g$};
\path node (c4) [above right=of h0] {$\tilde{c}_4$};
\path node (c5) [right=of c4] {$\tilde{c}_5$};
\path node (ret) [above right=of c5] {};

\path[draw,->] (caller) -- (c0);
\path[draw,->] (c0) -- (c1);
\path[draw,->] (c1) -- (g0.west);
\path[draw,->] (g0.east) -- (c2);
\path[draw,->] (c2) -- (c3);
\path[draw,->] (c3) -- (h0.west);
\path[draw,->] (h0.east) -- (c4);
\path[draw,->] (c4) -- (c5);
\path[draw,->] (c5) -- (ret);

\path[draw,->,dashed] (c1) -- (c2);
\path[draw,->,dashed] (c3) -- (c4);
\end{tikzpicture}
\end{center}
The configuration $\tilde{c}_0$ represents the entry point to the function, and its incoming edge is a call-site transition.
The configuration $\tilde{c}_5$ represents an exit point for the function, and its outgoing edge is a return-point transition.
A transition where one intraprocedural configuration follows another, like $\tilde{c}_0 \rightarrow \tilde{c}_1$, is not technically possible in our
restricted ANF language but in more general languages would be.
The function's body may call other functions $f$ and $g$ whose configurations are not a part of the same intraprocedural set of nodes.
The primary insight behind our technique is that \textit{a set of intraprocedural configurations (like $\tilde{c}_0$ through $\tilde{c}_5$) 
necessarily share the exact same set of genuine continuations} (in this example, the incoming call-sites for $\tilde{c}_0$).

We call the set of configurations $\tilde{c}_0$ through $\tilde{c}_5$ an \textit{intraprocedural group} because they are those configurations that represent the body
of a function for a single abstract invocation---defined by an entry point unique to some $e$ and $\tilde{\rho}$.
Our central insight is to notice that this idea of an intraprocedural group also corresponds to those configurations that share a single set of continuations.
Our finite-state machine represents this set of continuations with a continuation address, so if this continuation address is precise enough to uniquely determine an 
intraprocedural group's entry point ($e$ and $\tilde{\rho}$), then it can be used for all configurations in that same group. 
Thus our allocator may be defined as simply:
\begin{align*}
  \widetilde{\var{alloc}}_{\cont\ \!\scalebox{0.5}{\text{P4F}}}((\expr, \tilde{\rho}, \tilde{\sigma}, \tilde{\sigma}_\cont, \tilde{a}_\cont), \expr', \tilde{\rho}', \tilde{\sigma}') 
  = (\expr', \tilde{\rho}')
\end{align*}
The impact of this change is easily missed, belied by its simplicity.
We allocate a continuation based only on the expression and environment at the entry point of each intraprocedural sequence of \texttt{let}-forms and it is precisely reinstated
when each of the calls in these \texttt{let}-forms return.

Recall that the monovariant continuation allocator in our example from section~\ref{sec:background:imprecision} resulted in return-flow merging 
because a single continuation address was being used for transitions to multiple entry points of different intraprocedural groups. 
More generally, return-flow merging occurs in a finite-state analysis when, at some return-point configuration $(\aexpr, \tilde{\rho}_{\aexpr}, \tilde{a}_\cont)$, the set 
of continuations for $\tilde{a}_\cont$ is less precise than the set of source configurations that transition to the entry point $(\expr, \tilde{\rho})$ 
of the same intraprocedural group. 
Because we allocate a continuation address specific to this exact entry point, and because that address is propagated by shallowly copying it to each return point for the same
intraprocedural group, the set of continuations will be as precise as the set of source configurations transitioning to the same entry point in all cases.
This means the return-flow merging problem cannot occur when using $\widetilde{\var{alloc}}_{\cont\ \!\scalebox{0.5}{\text{P4F}}}$
and neither is there a run-time overhead for stack introspection.

In section~\ref{sec:proof}, we formalize these intuitions and provide a proof that our unbounded-stack analysis simulates (i.e., is no more precise than) a finite-state analysis when using 
$\widetilde{\var{alloc}}_{\cont\ \!\scalebox{0.5}{\text{P4F}}}$.

\subsection{Complexity}
To see why this allocation scheme leads to only a constant-factor overhead,
consider a set of configurations $\tilde{c}_0,\tilde{c}_1,\cdots,\tilde{c}_n$ that form an intraprocedural group and
a set of call sites transitioning to $\tilde{c}_0$ with the continuations $\tilde{\kappa}_0,\tilde{\kappa}_1,\cdots,\tilde{\kappa}_{m-1}$.
We can diagrammatically visualize this as the following.
\begin{center}
\begin{tikzpicture}[node distance=0.8cm]
\path node (k1) {};
\path node (k2) [below=of k1] {};
\path node (k3) [below=of k2] {};
\path node (k4) [below=of k3] {};

\path node (c0) [right=of k2] {$\tilde{c}_0$};
\path node (c1) [right=of c0] {$\tilde{c}_1$};
\path node (c2) [right=of c1] {$\tilde{c}_{n-1}$};
\path node (c3) [right=of c2] {$\tilde{c}_n$};

\path node (r2) [right=of c3] {};
\path node (r1) [above=of r2] {};
\path node (r3) [below=of r2] {};
\path node (r4) [below=of r3] {};

\path[draw,->] (k1) -- node [sloped,above] {$\tilde{\kappa}_0$} (c0);
\path[draw,->] (k2) -- node [sloped,above] {$\tilde{\kappa}_1$} (c0);
\path[draw,->] (k4) -- node [sloped,above] {$\tilde{\kappa}_{m-1}$} (c0);
\path (k2) +(0.4cm,-0.3cm) node {$\vdots$};

\path[draw,->] (c0) -- (c1);
\path[] (c1) -- node{$\cdots$} (c2);
\path[draw,->] (c2) -- (c3);

\path[draw,->] (c3) -- node [sloped,above] {$\tilde{\kappa}_0$} (r1);
\path[draw,->] (c3) -- node [sloped,above] {$\tilde{\kappa}_1$} (r2);
\path[draw,->] (c3) -- node [sloped,above] {$\tilde{\kappa}_{m-1}$} (r4);
\path (r2) +(-0.4cm,-0.3cm) node {$\vdots$};
\end{tikzpicture}
\end{center}
Note that, for each call site, there is a corresponding return flow using the same continuation.
Our allocation strategy means that all of the configurations $\tilde{c}_0,\tilde{c}_1,\cdots,\tilde{c}_n$ use
the same continuation address $(\expr,\tilde{\rho})$.
The global continuation store then maps this address to the set $\{\tilde{\kappa}_0,\tilde{\kappa}_1,\cdots,\tilde{\kappa}_{m-1}\}$.

Now consider what must be done if a new call site transitions to $c_0$.
First, the continuation store must be extended to contain the continuation for this new call site, say $\tilde{\kappa}_{m}$,
in the continuation set at the address $(\expr,\tilde{\rho})$.
Then the corresponding return edge transitions must be added.
Note that none of $\tilde{c}_0,\tilde{c}_1,\cdots,\tilde{c}_{n-1}$ need to be modified or accessed.
The only work done here beyond that of the underlying analysis is the extension of the continuation store by
adding $\tilde{\kappa}_{m}$ at $(e, \tilde{\rho})$ and the addition of a corresponding return edge at return points.
Thus, the additional work is a constant factor of the number of times a continuation is added to the continuation store.

A na\"ive analysis might lead us to conclude that this is bounded by the product 
of the number of continuation addresses and the number of continuations.
However, there is a tighter bound.
Each transition adds only one continuation to the continuation store.
Thus the work done is a constant factor of the number of transitions in the underlying analysis.

Note that this differs from AAC, which may make duplicate copies of the continuation set for an intraprocedural group as it produces one for each combination of 
components $\expr$, $\tilde{\rho}$, and $\tilde{\sigma}$ drawn from the source states transitioning to it.
As a consequence, AAC allocates addresses strictly more unique than the target $(\expr', \tilde{\rho}')$ configuration.
Two different source expressions $\expr_0$ 
and $\expr_1$ may both have transitions to $(\expr', \tilde{\rho}')$, but AAC will produce two different target configurations $\tilde{c}'_0$ and $\tilde{c}'_1$ 
because the continuation addresses they allocate will be distinct.
This difference is maintained through the two variants of the function starting at $\expr'$ with environment $\tilde{\rho}'$, and when an 
exit point $\aexpr$ is reached for each, the expression and its environment are the same and propagate the same values to two sets of continuations.
Thus, these continuation addresses and the sets of stacks they represent are kept separate without any benefit.

PDCFA, on the other hand, is more complex for an entirely different reason: the epsilon closure graph.
Without the epsilon closure graph, PDCFA has no way to efficiently determine a topmost stack frame at each return transition.
Both our method and AAC's method make this trivial by propagating an address explicitly to each state.
While our method allows a continuation address to be shallowly propagated across each intraprocedural node in a function, the epsilon closure graph recomputes 
a separate set of incoming epsilon edges for every node.
This means that the number of such edges for a given entry point $(\expr, \hat{\rho})$ is the number of callers times the number of intraprocedural nodes.
This is a quadratic blow up from the number of nodes in a finite-state model.
This is why monovariant store-widened PDCFA is in $O(n^6)$ instead of in $O(n^3)$ like traditional $0$-CFA.
We are able to naturally exploit our insight that each intraprocedural node following an entry point $(\expr, \hat{\rho})$ shares the same set of continuations 
(i.e., the same epsilon edges) by propagating a pointer to this set instead of rebuilding it for each node.
PDCFA is unable to exploit this insight without adding machinery to propagate only a shallow copy of an incoming epsilon edge set intraprocedurally.
It is likely that this insight could also be imported into the PDCFA style of analysis to yield a variant of PDCFA that incurs only a constant-factor overhead, 
but this would require additional machinery.

\subsection{Constant Overhead Requires Store Widening}
That no function can have two entry points that lead to the same exit point is a genuine restriction worth discussing further.
If this were not true, our technique would be precise (assuming multiple entry points are not merged), but it would not 
necessarily be a constant-factor increase in complexity.
The combination of no store widening (per-state value stores) and mutation is a good example of how this situation could arise.

To see how per-state stores can cause a further blow up in complexity, consider a function that is called with two different continuations 
and two different stores.
Without store widening, each store causes a different state to be created for the entry point $(e,\tilde{\rho})$ of the function. 
In the following diagram for example, $\tilde{\varsigma}_1$ is the state for the entry point with one store and $\tilde{\varsigma}'_1$ the state
for the entry point with another store.
\begin{center}
\begin{tikzpicture}[node distance=0.8cm]
\path node (k0) {};
\path node[below right=of k0] (c1) {$\tilde{\varsigma}_1$};
\path node[below left=of c1] (k1) {};
\path node[below right=of k1] (d1) {$\tilde{\varsigma}'_1$};

\path node[right=of c1] (c2) {$\tilde{\varsigma}_2$};
\path node[below right=of c2] (f1) {};
\path node[above right=of f1] (c3) {$\tilde{\varsigma}_3$};
\path node[right=of c3] (c4) {$\tilde{\varsigma}_4$};
\path node[above right=of c4] (r0) {};

\path node[right=of d1] (d2) {$\tilde{\varsigma}'_2$};
\path node[below right=of d2] (f2) {$f$};
\path node[above right=of f2] (d3) {$\tilde{\varsigma}'_3$};
\path node[right=of d3] (d4) {$\tilde{\varsigma}'_4$};
\path node[above right=of d4] (r2) {};

\path[draw,->] (k0) -- node[sloped,above] {$\tilde{\kappa}_1$} (c1);
\path[draw,->] (c1) -- (c2);
\path[draw,->] (c2) -- (f2.north west);
\path[draw,->] (f2.east) -- (c3);
\path[draw,->] (c3) -- (c4);
\path[draw,->] (c4) -- node[sloped,above] {} (r1);

\path[draw,->] (k1) -- node[sloped,above] {$\tilde{\kappa}_2$} (d1);
\path[draw,->] (d1) -- (d2);
\path[draw,->] (d2) -- (f2.west);

\path[draw,->,dashed] (f2.east) -- (d3);
\path[draw,->,dashed] (d3) -- (d4);
\path[draw,->,dashed] (d4) -- node[sloped,above] {} (r2);
\end{tikzpicture}
\end{center}
Now suppose that along both sequences of states there is a call to some function $f$ and that $f$ contains a side effect that causes the previously 
different value stores to become equal. 
For example, in $\tilde{\varsigma}_1$ perhaps the address for $x$ maps to $\{\texttt{\#t}\}$ and in $\tilde{\varsigma}_1'$ it maps to $\{\texttt{\#f}\}$.
If $x$ becomes bound to $\{\texttt{\#t},\texttt{\#f}\}$ along both paths
in the body of $f$, the stores along both paths would become identical.

A problem now arises.
Should $f$ return only to one state using this common store such as $\tilde{\varsigma}_3$ or should
it return to two different states (with identical stores) such as both $\tilde{\varsigma}_3$ and $\tilde{\varsigma}'_3$?
Either choice has drawbacks.
The semantics we have given would naturally yield the latter option, producing two distinct states that differ only by their continuation addresses (their original entry point).
Because these states are otherwise identical, splitting $\tilde{\cont}_1$ and $\tilde{\cont}_2$ into sets represented by two different continuation addresses results in
additional transitions and complexity without any benefit.
Arguably, these continuation sets should be merged and represented by a single address.
This corresponds to the former option and could save on run-time complexity but only at the cost of additional analysis machinery.
This means per-state stores are incompatible with our goal of obtaining perfect stack precision \textit{for free} in both senses (running time and human labor). 

\begin{figure}[t!]
  \centering
\begin{tikzpicture}
  \begin{axis}[
    ybar=1pt,
    ymin=0,
    width=\columnwidth,
    bar width=2.5pt,
    legend cell align=left,
    legend pos=north west,
    legend columns=2,
    xtick=data,
    x tick label style={rotate=90,anchor=east},
    symbolic x coords={mj09, eta, kcfa2, kcfa3, blur, loop2, sat, ack, cpstak, tak},
  ]
  \addplot[fill=black]
    coordinates
    {(mj09,28) (eta,28) (kcfa2,39) (kcfa3,55) (blur,123) (loop2,48) (sat,84) (ack,46) (cpstak,59) (tak,76)};
  \addplot[draw=black]
    coordinates
    {(mj09,14) (eta,15) (kcfa2,11) (kcfa3,13) (blur,19) (loop2,14) (sat,22) (ack,10) (cpstak,14) (tak,10)};
  \addplot[pattern=crosshatch]
    coordinates
    {(mj09,54) (eta,55) (kcfa2,100) (kcfa3,155) (blur,476) (loop2,72) (sat,311) (ack,143) (cpstak,109) (tak,384)};
  \addplot[pattern=north east lines]
    coordinates
    {(mj09,28) (eta,30) (kcfa2,36) (kcfa3,47) (blur,62) (loop2,25) (sat,52) (ack,17) (cpstak,32) (tak,24)};
  \legend{
    AAC Configurations,
    P4F Configurations,
    AAC States,
    P4F States
  }
  \end{axis}
\end{tikzpicture}
  \caption{A monovariant comparison.}
  \label{fig:0cfa} 
\end{figure}

\subsection{Implementation}
We have implemented both our technique and AAC's technique for analysis of a simplified Scheme intermediate language.
This language extends $\syn{Exp}$ with a variety of additional core forms including conditionals, mutation, recursive binding, 
tail calls, and a library of primitive operations.
Our implementation was written in Scala and executed using Scala 2.11 for OSX on an Intel Core i5 (1.3 GHz) with 4GB of RAM.
It is built upon the implementation of Earl et al.~\cite{earl:2012:introspectivepushdown}, which implements both traditional $k$-CFA and PDCFA.
The test cases we ran came from the Larceny R6RS benchmark suite (ack, cpstak, tak) and examples compiled from the previous literature 
on obtaining perfect stack precision (mj09, eta, kcfa2, kcfa3, blur, loop2, sat).
As a sanity check, we have verified that both AAC and our method produce results of equivalent precision in every case.
We ran each comparison using both a monovariant value-store allocator (figure~\ref{fig:0cfa}), and a $1$-call-sensitive polyvariant allocator (figure~\ref{fig:1cfa}).
Across the board, our method requires visiting strictly fewer machine configurations.
In some of these cases the difference is rather small, but in others it is significant.
We saw as much as a $16.0\times$ improvement in the monovariant analysis and as much as
a $10.4\times$ improvement in the context-sensitive analysis.
The mean speedup in terms of states visited was $5.4\times$ and $4.9\times$ in the monovariant and context-sensitive analyses, respectively.

\section{Proof of Precision}
\label{sec:proof}
Proving soundness is fairly straight forward as discussed in section~\ref{sec:background:soundness}, but proving precision poses a greater challenge.
To do this, we first define a simulation relation $(\simXi)$ where $\hat{\xi}
\simXi \tilde{\xi}$ (read as ``$\hat{\xi}$ simulates $\tilde{\xi}$'')
if and only if all stored values and machine configurations in
$\tilde{\xi}$ (including stacks implicit in this configuration) are
accounted for in the unbounded-stack representation $\hat{\xi}$.
Usually, the next step in such a proof would be to show that taking parallel steps preserves precision as in fallacy~\ref{fallacy:1}.
\begin{fallacy}[Steps preserve precision]
\label{fallacy:1}
If $\hat{\xi} \pdTo \hat{\xi}'$ and $\tilde{\xi} \fsTo \tilde{\xi}'$, then
$\hat{\xi} \simXi \tilde{\xi}$ implies $\hat{\xi}' \simXi \tilde{\xi}'$.
\end{fallacy}
However, fallacy~\ref{fallacy:1} is not true.
This is because after some finite number of steps $\tilde{\xi}$ may contain a cycle in its continuation store.
This means that an infinite family of successively longer stacks must also be in $\hat{\xi}$ for precision to hold.
After a finite number of steps, however, all stacks in $\hat{\xi}$ are bounded by a finite length.
Hence, there are stacks that precision says should be in $\hat{\xi}$ that are not.

We thus take a different approach to proving precision.
Before going into the details, the high-level overview of this proof is as follows.
Instead of stepping both $\hat{\xi}$ and $\tilde{\xi}$ in parallel,
we show that successive steps of $\tilde{\xi}$ are all precise
relative to any $\hat{\xi}$ that is already at a fixed point (i.e.,
theorem~\ref{theorem:step_prec_xi} found at the end of this section).
To show this, we need two inductions.
One is over the steps taken by $\tilde{\xi}$, and the other is over the stacks implied by $\tilde{\xi}$.
To separate these inductions, we define a well-formedness property ($\wf$ in figure~\ref{figure:wf}) that we can show is preserved by 
iterative steps from an initial $\tilde{\xi}_0$ (lemma~\ref{lemma:iterated_wf}) and for which we can show that any well-formed $\tilde{\xi}$ 
is precise relative to any $\hat{\xi}$ that is at a fixed point (lemmas~\ref{lemma:step_prec_Store} and~\ref{lemma:step_prec_R}).

The well-formedness property is defined in terms of two additional concepts.
First, we formally define the stacks, $\tilde{\psi}$, implied by a continuation address, $\tilde{a}_\cont$, and continuation store, $\tilde{\sigma}_\cont$, 
in terms of a relation $\tilde{\psi} \inpsi \tilde{a}_\cont\ \text{(via $\tilde{\sigma}_\cont$)}$ that we define. 
Then, we define paths, $(\pdPath)$ and $(\fsPath)$, through $\hat{\xi}$ and $\tilde{\xi}$ in terms of a sequence of state steps, $(\pdTo)$ and $(\fsTo)$,
between states represented by configurations in $\hat{\xi}$ and $\tilde{\xi}$. 
This allows us to prove the precision of any well-formed $\tilde{\xi}$ (i.e., lemma~\ref{lemma:step_prec_R}) through
a logical chain informally shown in figure~\ref{figure:logical-chain}.
\begin{figure}[t!]
  \centering
\begin{tikzpicture}
  \begin{axis}[
    ybar=1pt,
    ymin=0,
    width=\columnwidth,
    bar width=2.5pt,
    legend cell align=left,
    legend pos=north west,
    legend columns=2,
    xtick=data,
    x tick label style={rotate=90,anchor=east},
    symbolic x coords={mj09, eta, kcfa2, kcfa3, blur, loop2, sat, ack, cpstak, tak},
  ]
  \addplot[fill=black]
    coordinates
    {(mj09,45) (eta,23) (kcfa2,139) (kcfa3,253) (blur,93) (loop2,205) (sat,407) (ack,124) (cpstak,46) (tak,214)};
  \addplot[draw=black]
    coordinates
    {(mj09,23) (eta,16) (kcfa2,30) (kcfa3,46) (blur,36) (loop2,33) (sat,69) (ack,28) (cpstak,17) (tak,34)};
  \addplot[pattern=crosshatch]
    coordinates
    {(mj09,94) (eta,39) (kcfa2,432) (kcfa3,1113) (blur,182) (loop2,354) (sat,1998) (ack,629) (cpstak,73) (tak,1387)};
  \addplot[pattern=north east lines]
    coordinates
    {(mj09,55) (eta,26) (kcfa2,135) (kcfa3,275) (blur,75) (loop2,68) (sat,199) (ack,82) (cpstak,27) (tak,133)};
  \legend{
    AAC Configurations,
    P4F Configurations,
    AAC States,
    P4F States
  }
  \end{axis}
\end{tikzpicture}
  \caption{A $1$-call-sensitive comparison.}
   \label{fig:1cfa}
\end{figure}
\begin{figure}
\noindent\begin{tikzpicture}[->,>=stealth]
\node (ul) {$(\tilde{c}_0, \epsilon) \fsPath ((e, \tilde{\rho}, \tilde{a}_\cont), \tilde{\psi})$};
\node (ur) [right of=ul,node distance=15em] {$\begin{aligned} &(e, \tilde{\rho}, \tilde{a}_\cont) \in \tilde{r} \\ &\tilde{\psi} \inpsi \tilde{a}_\cont\ \text{(via $\tilde{\sigma}_\cont$)} \end{aligned}$};
\node (ll) [below of=ul,node distance=5em] {$\hat{c}_0 \pdPath (e, \hEnv(\tilde{\rho}), \hKont(\tilde{\psi}))$};
\path (ur |- ll) node (lr) {$(e, \hEnv(\tilde{\rho}), \hKont(\tilde{\psi})) \in \hat{r}$};

\path [draw] (ur) -- node [above] {\small{Lemma~\ref{lemma:stacks_have_paths}}} (ul);
\path [draw] (ul) -- node [right] {\small{Lemma~\ref{lemma:path_til_hat}}} (ll);
\path [draw] (ll) -- node [above] {\small{Lemma~\ref{lemma:path_end_point}}} (lr);
\path [draw, dashed] (ur) -- node [right] {\small{Lemma~\ref{lemma:step_prec_R}}} (lr);
\end{tikzpicture}
\caption{The logical chain proving lemma~\ref{lemma:step_prec_R}.  Assumes $(\hat{r}, \hat{\sigma})$ is at a fixed point and 
$(\tilde{r}, \tilde{\sigma}, \tilde{\sigma}_\cont)$ is well-formed.}
\label{figure:logical-chain}
\end{figure}
In lemma~\ref{lemma:stacks_have_paths}, we show that for any configuration $(e, \tilde{\rho}, \tilde{a}_\cont)$ in the $\tilde{r}$ of a $\tilde{\xi}$ and any $\tilde{\psi}$ implied by $a_\cont$ with the continuation store $\sigma_\cont$ of $\tilde{\xi}$, there exists a path from the initial configuration $\tilde{c}_0$ with an empty stack $\epsilon$ to the configuration $(e, \tilde{\rho}, \tilde{a}_\cont)$ with the implied stack $\tilde{\psi}$.
In lemma~\ref{lemma:path_til_hat}, we then show that there exists a corresponding path in $\hat{\xi}$ from $\hat{c}_0$ to $(e, \hEnv(\tilde{\rho}), \hKont(\tilde{\psi}))$.
Finally, in lemma~\ref{lemma:path_end_point}, we show that the endpoint of that path is in $\hat{\xi}$ and thus the set of reachable configurations in $\tilde{\xi}$ is precise relative 
to $\hat{\xi}$ (lemma~\ref{lemma:step_prec_R}).

\subsection{Definitions and Assumptions}
In order to prove precision, we first require that the address spaces for both $\hat{\sigma}$ and $\tilde{\sigma}$ correspond as follows.
\begin{assumption}[Address equivalence] There exists an equivalence $(\equiv_{Addr})$ between finite-state-machine addresses $(\widetilde{\var{Addr}})$ 
and unbounded-stack-machine addresses $(\sa{Addr})$ that can be decomposed into a bijection $\sa{Addr}\ \underset{\!\!\scalebox{0.6}[0.6]{$H_{\scalebox{0.7}[0.7]{$\var{Addr}$}}$}}{\overset{\!\!\scalebox{0.6}[0.6]{$T_{\scalebox{0.7}[0.7]{$\var{Addr}$}}$}}{\rightleftharpoons}}\ \widetilde{\var{Addr}}$.
\label{simulation-alloc}
\end{assumption}

\begin{figure*}
\begin{align*}
  ((e, \tilde{\rho}, \tilde{a}_\cont), \tilde{\psi}) &\fsPath ((e, \tilde{\rho}, \tilde{a}_\cont), \tilde{\psi})\ \ \text{(via $(\tilde{r}, \tilde{\sigma}, \tilde{\sigma}_\cont)$, $(\tilde{r}', \tilde{\sigma}', \tilde{\sigma}'_\cont)$)}, \text{where}
  \\
  (e, \tilde{\rho}, \tilde{a}_\cont) &\in \tilde{r}'
  \qquad\qquad
  \tilde{\psi} \inpsi \tilde{a}_\cont\ \ \text{(via $\tilde{\sigma}'_\cont$)}
\end{align*}
\medskip
\begin{align*}
  ((e, \tilde{\rho}, \tilde{a}_\cont), \tilde{\psi}) &\fsPath ((e', \tilde{\rho}_\cont', \tilde{a}_\cont'), \tilde{\psi}')\ \ \text{(via $(\tilde{r}, \tilde{\sigma}, \tilde{\sigma}_\cont)$, $(\tilde{r}', \tilde{\sigma}', \tilde{\sigma}'_\cont)$)}, \text{where}
  \\
  \tilde{\rho}_\cont' &= \tilde{\rho}_\cont[x \mapsto \widetilde{alloc}(x, (\aexpr, \tilde{\rho}'', \tilde{\sigma}, \tilde{\sigma}_\cont, \tilde{a}_\cont''))]
  \\
  (\aexpr, \tilde{\rho}'', \tilde{a}_\cont'') &\in \tilde{r}
  \qquad\qquad
  (e', \tilde{\rho}_\cont', \tilde{a}_\cont') \in \tilde{r}'
  \qquad\qquad
  ((x, e', \tilde{\rho}_\cont), \tilde{a}_\cont') \in \tilde{\sigma}_\cont(\tilde{a}_\cont'')
  \\
  ((e, \tilde{\rho}, \tilde{a}_\cont), \tilde{\psi}) &\fsPath ((\aexpr, \tilde{\rho}'', \tilde{a}_\cont''), ((x, e', \tilde{\rho}_\cont), \tilde{a}_\cont') : \tilde{\psi}')\ \ \text{(via $(\tilde{r}, \tilde{\sigma}, \tilde{\sigma}_\cont)$, $(\tilde{r}', \tilde{\sigma}', \tilde{\sigma}'_\cont)$)}
  \\
  (\aexpr, \tilde{\rho}'', \tilde{\sigma}, \tilde{\sigma}_\cont, \tilde{a}_\cont'') &\fsToSub (e', \tilde{\rho}_\cont', \tilde{\sigma}', \tilde{\sigma}'_\cont, \tilde{a}_\cont')
\end{align*}
\medskip
\begin{align*}
  ((e, \tilde{\rho}, \tilde{a}_\cont), \tilde{\psi}) &\fsPath ((e', \tilde{\rho}', \tilde{a}_\cont'), ((x, e'', \tilde{\rho}''), \tilde{a}_\cont'') : \tilde{\psi}')\ \ \text{(via $(\tilde{r}, \tilde{\sigma}, \tilde{\sigma}_\cont)$, $(\tilde{r}', \tilde{\sigma}', \tilde{\sigma}'_\cont)$)}, \text{where}
  \\
  ((e, \tilde{\rho}, \tilde{a}_\cont), \tilde{\psi}) &\in \tilde{r}
  \qquad\qquad
  ((x, e'', \tilde{\rho}''), \tilde{a}_\cont'') \in \tilde{\sigma}'_\cont(\tilde{a}_\cont')
  \\
  ((e, \tilde{\rho}, \tilde{a}_\cont), \tilde{\psi}) &\fsPath ((\letiform{x}{\appform{f}{\aexpr}}{e''}, \tilde{\rho}'', \tilde{a}_\cont''), \tilde{\psi}')\ \ \text{(via $(\tilde{r}, \tilde{\sigma}, \tilde{\sigma}_\cont)$, $(\tilde{r}', \tilde{\sigma}', \tilde{\sigma}'_\cont)$)}
  \\
  (\letiform{x}{\appform{f}{\aexpr}}{e''}, \tilde{\rho}'', \tilde{\sigma}, \tilde{\sigma}_\cont, \tilde{a}_\cont'') &\fsToSub (e', \tilde{\rho}', \tilde{\sigma}', \tilde{\sigma}'_\cont, \tilde{a}_\cont')
\end{align*} 
\caption{Finite-state paths}
\label{figure:finite-state-paths}
\end{figure*}

From this we can build up equivalences ($\simEnv$, $\simFrame$, $\simClo$) and precision relations ($\simXi$, $\simStore$, $\simR$, $\simD$, $\simStore$) for all the components of our machine.
In addition, we can define conversion bijections ($\hEnv$, $\tEnv$, $\hFrame$, $\tFrame$, $\hClo$, $\tClo$, $\hD$, $\tD$, $\hStore$, $\tStore$) for most but not
all of the components.
These relations have the following signatures.
\begin{align*}
  (\simXi) &\subseteq \hat{\Xi} \times \tilde{\Xi}  && \text{[state-space precision]}
  \\
  (\simStore) &\subseteq \sa{Store} \times \widetilde{\var{Store}}  && \text{[store precision]}
  \\
  (\simR) &\subseteq \hat{R} \times \tilde{R} \times \widetilde{\var{KStore}} && \text{[reachable configs.\ precision]} 
  \\
  (\simEnv) &\subseteq \sa{Env} \times \widetilde{\var{Env}}  && \text{[env.\ equivalence]}
  \\
  (\simFrame) &\subseteq \sa{Frame} \times \widetilde{\var{Frame}}  && \text{[frame equivalence]}
  \\
  (\simD) &\subseteq \hat{D} \times \tilde{D}  && \text{[flow-set precision]}
  \\
  (\simClo) &\subseteq \sa{Clo} \times \widetilde{\var{Clo}}  && \text{[closure equivalence]}
\end{align*}

In addition, with the following assumption, we require that the value allocators respect the address correspondence.
\begin{assumption}[Allocation equivalence]
\label{assumption:prec_alloc}
If $\hat{\rho} \simEnv \tilde{\rho}$, $\hat{\sigma} \simStore \tilde{\sigma}$, and 
$\tilde{\psi} \inpsi \tilde{a}_\cont\ \text{(via $\tilde{\sigma}_\cont$)}$, then:
\begin{align*}
  \sa{alloc}(x, (e, \hat{\rho}, \hat{\sigma}, \hKont(\tilde{\psi}))) 
  \equiv_{Addr}   
  \widetilde{\var{alloc}}(x, (e, \tilde{\rho}, \tilde{\sigma}, \tilde{\sigma}_\cont, \tilde{a}_\cont))
\end{align*}
\end{assumption}
This assumption uses $(\inpsi)$ and $\hKont$, which deal with the stacks implied by an address and continuation store.
We define an implied stack as an unbounded list of finite-state continuations $\tilde{\cont}$:
\begin{align*}
  \tilde{\psi} \in \tilde{\Psi} = \widetilde{\var{Kont}}^{*} && \text{[implied stack]} 
\end{align*}
These $\tilde{\psi}$ are an intermediate representation in that, like $\hat{\cont}$, their structure is unbounded, but each element is taken
directly from the finite-state machine.
We define a trinary relation $(\inpsi)$ that specifies which $\tilde{\psi}$ are implied by an $\tilde{a}_\cont$ in $\tilde{\sigma}_\cont$. 
This has the following base case and inductive case:
\begin{align*}
  \epsilon &\inpsi \tilde{a}_{\text{halt}}\ \text{(via $\tilde{\sigma}_\cont$)}\\
\end{align*}
\begin{align*}
  (\tilde{\phi}, \tilde{a}_\cont') \in \tilde{\sigma}_\cont(\tilde{a}_\cont&) \wedge \tilde{\psi} \inpsi \tilde{a}_\cont'\ \text{(via $\tilde{\sigma}_\cont$)} \wedge \tilde{a}_\cont \neq \tilde{a}_\text{halt}
  \\ &\ \ \ \ \ \implies
  ((\tilde{\phi}, \tilde{a}_\cont'):\tilde{\psi}) \inpsi \tilde{a}_\cont\ \text{(via $\tilde{\sigma}_\cont$)}
\end{align*}

Then given such a $\tilde{\psi}$, we can directly construct its equivalent unbounded stack:
\begin{align*}
  \hKont(\epsilon) &\defas \epsilon &&&
  \hKont((\tilde{\phi}, \tilde{a}_\cont) : \tilde{\psi}) &\defas \hFrame(\tilde{\phi}) : \hKont(\tilde{\psi})
\end{align*}
Also, given a finite-state configuration $\tilde{c}$ and an implicit stack $\tilde{\psi}$, we can construct an unbounded-stack configuration $\hat{c}$:
\begin{align*}
  \hC((e, \tilde{\rho}, \tilde{a}_\cont), \tilde{\psi}) \defas (e, \hEnv(\tilde{\rho}), \hKont(\tilde{\psi}))
\end{align*}

\begin{figure}
\begin{align*}
  (e, \hat{\rho}, \hat{\cont}) &\pdPath (e, \hat{\rho}, \hat{\cont})\ \ \text{(via $\hat{r}$, $\hat{\sigma}$)}
\end{align*} 
\vspace{-0.75cm}
\center{\rule{0.3\textwidth}{0.3pt}}
\begin{align*}
  (e, \hat{\rho}, \hat{\cont}) &\pdPath (e', \hat{\rho}', \hat{\cont}')\ \ \text{(via $\hat{r}$, $\hat{\sigma}$)}, \text{where}
  \\
  (e, \hat{\rho}, \hat{\cont}) &\pdPath (e'', \hat{\rho}'', \hat{\cont}'')\ \ \text{(via $\hat{r}$, $\hat{\sigma}$)}
  \\
  (e'', \hat{\rho}'', \hat{\sigma}, \hat{\cont}'') &\pdToSub (e', \hat{\rho}', \hat{\sigma}, \hat{\cont}')
\end{align*} 
\caption{Unbounded-stack paths}
\label{label:unbounded-stack-paths}
\end{figure}

Next, in figures~\ref{figure:finite-state-paths} and~\ref{label:unbounded-stack-paths} , we define paths to configurations.
For $(\pdPath)$ this is defined by a base case from a configuration to
itself and a recursive case that builds onto an existing path with a
step within $\hat{\xi}$.
This uses a variation of the step relation defined in figure~\ref{figure:sub-step-relations} that allows the output store of the step to be a sub-store of the store in $\hat{\xi}$.
The $(\fsPath)$ relation is similar but adds extra side conditions that ensure invariants used in our proof.

Then, in figure~\ref{figure:wf}, we define well-formedness.
This is a binary predicate with the first argument $\tilde{\xi}$ being the predecessor of the second argument $\tilde{\xi}'$,
which is the result we say is well-formed.
This predicate is defined in terms of several sub-properties.
The $\wf_{\tilde{\xi}}$ property requires $\tilde{\xi}$ be well-formed and the predecessor of $\tilde{\xi}'$.

The $\wf_{\sqsubseteq}$ property requires that $\tilde{\xi}'$ be component-wise greater than or equal to $\tilde{\xi}$.
The $\wf_\text{init}$ and $\wf_\text{halt}$ properties respectively require that the initial configuration be in $\tilde{\xi}'$ and
that the halt-continuation address $a_{\text{halt}}$ not have any continuations associated with it.
Finally, $\wf_{\tilde{r}}$, $\wf_{\tilde{\sigma}}$, and $\wf_{\tilde{\sigma}_\cont}$ ensure that everything in the $\tilde{r}$, $\tilde{\sigma}$, and $\tilde{\sigma}_\cont$ for $\tilde{\xi}'$ has a reason to be there.
For $\wf_{\tilde{r}}$, this means that every element of $\tilde{r}$ has some path leading to it.
For $\wf_{\tilde{\sigma}}$ and $\wf_{\tilde{\sigma}_\cont}$, this means that every value stored in $\tilde{\sigma}$ or $\tilde{\sigma}_\cont$ has some
step $(\fsTo)$ that put it there.
These are defined in terms of the variations of $(\fsToSub)$ in figure~\ref{figure:sub-step-relations} that have side conditions about the contents of the store.

\begin{figure}
\begin{align*}
  \wf(\tilde{\xi}, \tilde{\xi}') \defas \wf_{\tilde{\xi}}(\tilde{\xi}, \tilde{\xi}') \wedge \wf_{\sqsubseteq}(\tilde{\xi}, \tilde{\xi}') \wedge \wf_{\text{init}}(\tilde{\xi}, \tilde{\xi}') \wedge \wf_{\text{halt}}(\tilde{\xi}, \tilde{\xi}') \\ \wedge \wf_{\tilde{r}}(\tilde{\xi}, \tilde{\xi}') \wedge \wf_{\tilde{\sigma}}(\tilde{\xi}, \tilde{\xi}') \wedge \wf_{\tilde{\sigma}_\cont}(\tilde{\xi}, \tilde{\xi}')
\end{align*}
\center{\rule{0.3\textwidth}{0.3pt}}
\begin{align*}
  \wf_{\tilde{\xi}}(\tilde{\xi}, \tilde{\xi}') &\defas (\tilde{\xi} = \tilde{\xi}' = (\{(e_0, \varnothing, \tilde{a}_\text{halt})\}, \bot, \bot))
  \\
  &\vee\ 
  (\tilde{\xi} \fsToW \tilde{\xi}'
  \wedge
  \exists \tilde{\xi}''.\ \wf(\tilde{\xi}'', \tilde{\xi}))
\end{align*}
%
%
\begin{align*}
  \wf_{\sqsubseteq}((\tilde{r}, \tilde{\sigma}, \tilde{\sigma}_\cont), (\tilde{r}', \tilde{\sigma}', \tilde{\sigma}_\cont')) &\defas
  (\tilde{r} \subseteq \tilde{r}') \wedge (\tilde{\sigma} \sqsubseteq \tilde{\sigma}') \wedge (\tilde{\sigma}_\cont \sqsubseteq \tilde{\sigma}_\cont') 
\end{align*}
%
%
\begin{align*}
  \wf_\text{init}(\tilde{\xi}, (\tilde{r}'&, \tilde{\sigma}', \tilde{\sigma}_\cont')) \defas (e_0, \varnothing, \tilde{a}_\text{halt}) \in \tilde{r}'
\end{align*}
%
%
\begin{align*}
  \wf_\text{halt}(\tilde{\xi}, (\tilde{r}', \tilde{\sigma}', \tilde{\sigma}_\cont')) \defas \forall \tilde{\cont}.\ \tilde{\cont} \notin \tilde{\sigma}_\cont'(\tilde{a}_\text{halt})
\end{align*}
%
%
\begin{align*}
  \wf_{\tilde{r}}(&\tilde{\xi}, (\tilde{r}', \tilde{\sigma}', \tilde{\sigma}_\cont')) \defas
  \\
  &\forall (e, \tilde{\rho}, \tilde{a}_\cont) \in \tilde{r}'.\ \exists \tilde{\psi} \inpsi \tilde{a}_\cont\ \text{(via $\tilde{\sigma}_\cont'$)}.\  
  \\
  &\ \ \ \ \ \ \ \ ((e_0, \varnothing, \tilde{a}_\text{halt}), \epsilon) \fsPath ((e, \tilde{\rho}, \tilde{a}_\cont), \tilde{\psi})\ \ \text{(via $\tilde{\xi}$, $(\tilde{r}', \tilde{\sigma}', \tilde{\sigma}_\cont')$)}
\end{align*}
%
%
\begin{align*}
  \wf_{\tilde{\sigma}}(&(\tilde{r}, \tilde{\sigma}, \tilde{\sigma}_\cont), (\tilde{r}', \tilde{\sigma}', \tilde{\sigma}_\cont')) \defas
  \\
  &\forall \tilde{a}.\ \forall \widetilde{clo} \in \tilde{\sigma}'(\tilde{a}).\ 
  \exists (e, \tilde{\rho}, \tilde{a}_\cont) \in \tilde{r}.\ \exists (e', \tilde{\rho}', \tilde{a}_\cont') \in \tilde{r}'.\ 
  \\
  &\ \ \ \ \ \ \ \ (e, \tilde{\rho}, \tilde{\sigma}, \tilde{\sigma}_\cont, \tilde{a}_\cont) \fsToSub (e', \tilde{\rho}', \tilde{\sigma}', \tilde{\sigma}_\cont', \tilde{a}_\cont')\ \ \text{(with $\tilde{a}$, $\widetilde{clo}$)}
\end{align*}
%
%
\begin{align*}
  \wf_{\tilde{\sigma}_\cont}((\tilde{r}, \tilde{\sigma}, \tilde{\sigma}_\cont), (\tilde{r}', \tilde{\sigma}', \tilde{\sigma}_\cont')) \defas 
  \ \ \ \ \ \ \ \ \ \ \ \ \ \ \ \ \ \ \ \ \ \ \ \ \ \ \ \ \ \ \ \ \ \ \ \ \ \ \ \ \ \ 
  \ \ \ \ \ \ \ \ \ \ \ \ \ \ \ \ \ \ \ \ \ \ \ \ \ \ \ \ \ \ \ \ \ \ \ \ \ \ \ \ \ \ 
\end{align*}
\vspace{-0.5cm}
\begin{align*}
  &\forall \tilde{a}_\cont'.\ 
  \forall ((x, e, \tilde{\rho}_\cont), \tilde{a}_\cont) \in \tilde{\sigma}_\cont'(\tilde{a}_\cont').\ 
  \\
  &\ \ \ \ \ \ \exists e_\cont'.\ \exists \tilde{\rho}_\cont'.\ \tilde{a}_\cont'=(e_\cont', \tilde{\rho}_\cont')\wedge\mathit{Entry}(\tilde{a}_\cont') \in \tilde{r}'\ 
  \\
  &\ \ \ \ \ \ \ \ \ \ \ \wedge \exists f.\ \exists \aexpr.\ 
  (\letiform{x}{\appform{f}{\aexpr}}{e}, \tilde{\rho}_\cont, \tilde{a}_\cont) \in \tilde{r}
  \\
  &\ \ \ \ \ \ \ \ \ \ \ \ \ \ \ \ \ \wedge 
  (\letiform{x}{\appform{f}{\aexpr}}{e}, \tilde{\rho}_\cont, \tilde{\sigma}, \tilde{\sigma}_\cont, \tilde{a}_\cont) 
  \\
  &\ \ \ \ \ \ \ \ \ \ \ \ \ \ \ \ \ \ \ \ \ \ \ \fsToSub
  (e_\cont', \tilde{\rho}_\cont', \tilde{\sigma}', \tilde{\sigma}_\cont', \tilde{a}_\cont')\ \ 
  \text{(with $\tilde{a}_\cont'$, $((x, e, \tilde{\rho}_\cont), \tilde{a}_\cont)$)} 
\end{align*}
\center{\rule{0.3\textwidth}{0.3pt}}
\begin{align*}
  \var{Entry} :&\ \widetilde{\var{Addr}} \to \tilde{C}
  \\
  \var{Entry}(\tilde{a}_\text{halt}) \defas&\ (e_0, \varnothing, \tilde{a}_\text{halt})
  \\
  \var{Entry}((e_\cont, \tilde{\rho}_\cont)) \defas&\ (e_\cont, \tilde{\rho}_\cont, (e_\cont, \tilde{\rho}_\cont))
\end{align*}
\caption{Well-formedness properties}
\label{figure:wf}
\end{figure}

For $\wf_{\tilde{\sigma}_\cont}$, we define $\mathit{Entry}$, which maps a continuation address to the configuration that is the entry point for the function invocation 
that contains the configurations using that continuation address.

Finally, with the following assumption, we require that once allocation creates an address
it must always produce the same address for the same configuration even if the value or continuation
stores have changed.
\begin{assumption}[Allocation consistency] 
\label{assumption:alloc_consistency}
If $\wf(\tilde{\xi}, (\tilde{r}, \tilde{\sigma}, \tilde{\sigma}_\cont))$, and the state step $(e, \tilde{\rho}, \tilde{\sigma}, \tilde{\sigma}_\cont, \tilde{a}_\cont) \fsTo (e', \tilde{\rho}'[x \mapsto \tilde{a}], \tilde{\sigma}'', \tilde{\sigma}_\cont'', \tilde{a}_\cont')$ holds where $\tilde{a} = \widetilde{\var{alloc}}(x, (e, \tilde{\rho}, \tilde{\sigma}, \tilde{\sigma}_\cont, \tilde{a}_\cont))$ and there is a result step $(\tilde{r}, \tilde{\sigma}, \tilde{\sigma}_\cont) \fsToW (\tilde{r}', \tilde{\sigma}', \tilde{\sigma}_\cont')$, then the corresponding allocation for $e$, $\tilde{\rho}$, and $\tilde{a}_\cont$, but with $\tilde{\sigma}'$ and $\tilde{\sigma}_\cont'$, is the same:
\[ \widetilde{\var{alloc}}(x, (e, \tilde{\rho}, \tilde{\sigma}, \tilde{\sigma}_\cont, \tilde{a}_\cont)) = \widetilde{\var{alloc}}(x, (e, \tilde{\rho}, \tilde{\sigma}', \tilde{\sigma}_\cont', \tilde{a}_\cont)) \]
\end{assumption}

\begin{figure}
\begin{align*}
  (e, \hat{\rho}, \hat{\sigma}, \hat{\cont}) &\pdToSub (e', \hat{\rho}', \hat{\sigma}', \hat{\cont}'), \text{where}
  \\
  (e, \hat{\rho}, \hat{\sigma}, \hat{\cont}) &\pdTo (e', \hat{\rho}', \hat{\sigma}'', \hat{\cont}')
  \\
  \hat{\sigma}'' &\sqsubseteq \hat{\sigma}'
\end{align*} 
\begin{align*}
  (e, \tilde{\rho}, \tilde{\sigma}, \tilde{\sigma}_\cont, \tilde{a}_\cont) &\fsToSub (e', \tilde{\rho}', \tilde{\sigma}', \tilde{\sigma}_\cont', \tilde{a}_\cont'), \text{where}
  \\
  (e, \tilde{\rho}, \tilde{\sigma}'', \tilde{\sigma}_\cont'', \tilde{a}_\cont) &\fsTo (e', \tilde{\rho}', \tilde{\sigma}''', \tilde{\sigma}_\cont''', \tilde{a}_\cont')
  \\
  (\tilde{\sigma}'' \sqsubseteq \tilde{\sigma}) \wedge (\tilde{\sigma}_\cont'' \sqsubseteq \tilde{\sigma}_\cont) &\wedge 
  (\tilde{\sigma}''' \sqsubseteq \tilde{\sigma}') \wedge (\tilde{\sigma}_\cont''' \sqsubseteq \tilde{\sigma}_\cont')
\end{align*} 
\begin{align*}
  (e, \tilde{\rho}, \tilde{\sigma}, \tilde{\sigma}_\cont, \tilde{a}_\cont) &\fsToSub (e', \tilde{\rho}', \tilde{\sigma}', \tilde{\sigma}_\cont', \tilde{a}_\cont')
  \ \text{(with $\tilde{a}, \widetilde{clo}$)}, \text{where}
  \\
  (e, \tilde{\rho}, \tilde{\sigma}'', \tilde{\sigma}_\cont'', \tilde{a}_\cont) &\fsTo (e', \tilde{\rho}', \tilde{\sigma}''', \tilde{\sigma}_\cont''', \tilde{a}_\cont')
  \\
  (\tilde{\sigma}'' \sqsubseteq \tilde{\sigma}) \wedge (\tilde{\sigma}_\cont'' \sqsubseteq \tilde{\sigma}_\cont) &\wedge 
  (\tilde{\sigma}''' \sqsubseteq \tilde{\sigma}') \wedge (\tilde{\sigma}_\cont''' \sqsubseteq \tilde{\sigma}_\cont')
  \wedge \widetilde{clo} \in \tilde{\sigma}'''(\tilde{a})
\end{align*} 
\begin{align*}
  (e, \tilde{\rho}, \tilde{\sigma}, \tilde{\sigma}_\cont, \tilde{a}_\cont) &\fsToSub (e', \tilde{\rho}', \tilde{\sigma}', \tilde{\sigma}_\cont', \tilde{a}_\cont') 
  \ \text{(with $\tilde{\cont}, \tilde{a}''_\cont$)}, \text{where}
  \\
  (e, \tilde{\rho}, \tilde{\sigma}'', \tilde{\sigma}_\cont'', \tilde{a}_\cont) &\fsTo (e', \tilde{\rho}', \tilde{\sigma}''', \tilde{\sigma}_\cont''', \tilde{a}_\cont')
  \\
  (\tilde{\sigma}'' \sqsubseteq \tilde{\sigma}) \wedge (\tilde{\sigma}_\cont'' \sqsubseteq \tilde{\sigma}_\cont) &\wedge 
  (\tilde{\sigma}''' \sqsubseteq \tilde{\sigma}') \wedge (\tilde{\sigma}_\cont''' \sqsubseteq \tilde{\sigma}_\cont')
  \wedge \tilde{\cont} \in \tilde{\sigma}_\cont'''(\tilde{a}_\cont'')
\end{align*} 
\caption{Sub-step Relations}
\label{figure:sub-step-relations}
\end{figure}

\subsection{Lemmas and Theorems}

To start, lemma~\ref{lemma:iterated_wf} shows that iterated steps produce well-formed $\tilde{\xi}$. 

\begin{lemma}[Well-formedness of analysis results] 
\label{lemma:iterated_wf}
If $\tilde{\xi}'$ is the result of taking zero or more steps of $(\fsToW)$, starting from the initial result, $(\{(e_0, \varnothing, \tilde{a}_\text{halt})\}, \bot, \bot)$, then $\wf(\tilde{\xi},\tilde{\xi}')$ for some $\tilde{\xi}$.
\end{lemma}
\begin{proof}
We induct over $(\fsToW)$ steps. In the base case, we can easily show that the initial result is well-formed. 
We can also show that for any $\tilde{\xi}' \fsToW \tilde{\xi}''$, if $\wf(\tilde{\xi},\tilde{\xi}')$ then $\wf(\tilde{\xi}',\tilde{\xi}'')$.
This is done using sublemmas for the components of well-formedness.
We omit these for space.
\end{proof}

Next, with lemma~\ref{lemma:stacks_have_paths}, we show that every configuration paired with one of its implied stacks has a path leading to it (i.e., the top edge of figure~\ref{figure:logical-chain}).

\begin{lemma}[Stacks have paths]
\label{lemma:stacks_have_paths}
If $(e, \tilde{\rho}, \tilde{a}_\cont) \in \tilde{r}$ such that $\wf(\tilde{\xi}, (\tilde{r}, \tilde{\sigma}, \tilde{\sigma}_\cont))$, then:
\begin{align*}
  &\forall \tilde{\psi} \inpsi \tilde{a}_\cont\ \ \text{(via $\tilde{\sigma}_\cont$)}.\\
  &\ \ \ \ ((e_0, \varnothing, \tilde{a}_\text{halt}), \epsilon) \fsPath ((e, \tilde{\rho}, \tilde{a}_\cont), \tilde{\psi})\ \ \text{(via $\tilde{\xi}$, $(\tilde{r}, \tilde{\sigma}, \tilde{\sigma}_\cont)$)}
\end{align*}
\end{lemma}
\begin{proof}
By $\wf_{\tilde{r}}(\tilde{\xi}, (\tilde{r}, \tilde{\sigma}, \tilde{\sigma}_\cont))$, there is some 
$\tilde{\psi}' \inpsi \tilde{a}_\cont\ \ \text{(via $\tilde{\sigma}_\cont$)}$ for which 
$((e_0, \varnothing, \tilde{a}_\text{halt}), \epsilon) \fsPath ((e, \tilde{\rho}, \tilde{a}_\cont), \tilde{\psi}')\ \ \text{(via $\tilde{\xi}$, $(\tilde{r}, \tilde{\sigma}, \tilde{\sigma}_\cont)$)}$.
However, this path uses $\tilde{\psi}'$ instead of our desired $\tilde{\psi}$.
Thus we induct over $\tilde{\psi}$.
If $\tilde{\psi}$ is the empty list, $\epsilon$, then $\tilde{a}_\cont$ must be $\tilde{a}_\text{halt}$ and thus $\epsilon$ is the only $\tilde{\psi}$ for which 
$\tilde{\psi} \inpsi \tilde{a}_\cont\ \ \text{(via $\tilde{\sigma}_\cont$)}$.
So $\tilde{\psi}' = \tilde{\psi} = \epsilon$, and the path obtained from $\wf_{\tilde{r}}(\tilde{\xi}, (\tilde{r}, \tilde{\sigma}, \tilde{\sigma}_\cont))$ equals our desired conclusion.

If $\tilde{\psi}$ is $((x, e_\cont, \tilde{\rho}_\cont), \tilde{a}_\cont'):\tilde{\psi}''$ for some $x$, $e_\cont$, $\tilde{\rho}_\cont$, $\tilde{a}_\cont'$, 
$\tilde{\psi}''$, then there is a path for $\tilde{\psi}'$ from $\mathit{Entry}(\tilde{a}_\cont)$ to $(e, \tilde{\rho}, \tilde{a}_\cont)$.
By another induction there is a similar path for $\tilde{\psi}$:
\begin{align*}
  (\mathit{Entry}(\tilde{a}_\cont), \tilde{\psi}) \fsPath ((e, \tilde{\rho}, \tilde{a}_\cont), \tilde{\psi})\ \ \text{(via $\tilde{\xi}$, $(\tilde{r}, \tilde{\sigma}, \tilde{\sigma}_\cont)$)}
\end{align*}
By $\wf_{\tilde{r}}(\tilde{\xi}, (\tilde{r}, \tilde{\sigma}, \tilde{\sigma}_\cont))$, there exist $f$ and $\aexpr$ for a call site \linebreak
$(\letiform{x}{\appform{f}{\aexpr}}{e_\cont}, \tilde{\rho}_\cont, \tilde{a}_\cont') \in \tilde{r}$ and a step from that call site
to $\mathit{Entry}(\tilde{a}_\cont)$:
\begin{align*}
  (\letiform{x}{\appform{f}{\aexpr}}{e_\cont}, \tilde{\rho}_\cont, \tilde{\sigma}, \tilde{\sigma}_\cont, \tilde{a}_\cont')
  \fsToSub
  (e', \tilde{\rho}', \tilde{\sigma}, \tilde{\sigma}_\cont, \tilde{a}_\cont)
\end{align*}
\vspace{-0.5cm}
\begin{align*}
  \text{where}\ \ \ \ \ \ (e', \tilde{\rho}', \tilde{a}_\cont) = \mathit{Entry}(\tilde{a}_\cont)
\end{align*}
By the induction hypothesis, we have a path for $\tilde{\psi}''$ from $(e_0, \varnothing, \tilde{a}_\text{halt})$ to the call site:
\begin{align*}
  ((e_0, \varnothing, \tilde{a}_\text{halt}), \epsilon)
  \fsPath
  ((\letiform{x}{\appform{f}{\aexpr}}{e_\cont}, \tilde{\rho}_\cont, \tilde{a}_\cont'), \tilde{\psi}'')
\end{align*}
We now have a path from $((e_0, \varnothing, \tilde{a}_\text{halt}), \epsilon)$ to the call site with $\tilde{\psi}''$,
a step from the call site to $\mathit{Entry}(\tilde{a}_\cont)$ that pushes $(x, e_\cont, \tilde{\rho}_\cont)$ onto the stack,
and a path from $\mathit{Entry}(\tilde{a}_\cont)$ to $(e, \tilde{\rho}, \tilde{a}_\cont)$ with $\tilde{\psi}$.
From these we can then construct the path desired in our conclusion.
\end{proof}

Next, with lemma~\ref{lemma:path_til_hat}, we show that every path in a well-formed $\tilde{\xi}$ has a corresponding path in any $\hat{\xi}$ that is at at fixed point (i.e., the left edge of figure~\ref{figure:logical-chain}).

\begin{lemma}[Path conversion] 
\label{lemma:path_til_hat}
If $(e, \tilde{\rho}, \tilde{a}_\cont) \in \tilde{r}$ such that \linebreak
$\wf(\tilde{\xi}, (\tilde{r}, \tilde{\sigma}, \tilde{\sigma}_\cont))$
and $(\hat{r}, \hat{\sigma}) \pdToW (\hat{r}, \hat{\sigma})$, then:
\begin{align*}
  &((e_0, \varnothing, \tilde{a}_\text{halt}), \epsilon) \fsPath ((e, \tilde{\rho}, \tilde{a}_\cont), \tilde{\psi})\ \ \text{(via $\tilde{\xi}$, $(\tilde{r}, \tilde{\sigma}, \tilde{\sigma}_\cont)$)}
  \\
  &\ \ \ \ \ \ \ \ \implies (e_0, \varnothing, \epsilon) \pdPath \hC((e, \tilde{\rho}, \tilde{a}_\cont), \tilde{\psi})\ \ \text{(via $\hat{r}$, $\hat{\sigma}$)}
\end{align*}
\end{lemma}
\begin{proof}
By induction over the finite-state path. We have three cases.

\textbf{Case:} The path is empty. Trivial.

\textbf{Case:} The last step of the path is a return.
For some $\aexpr$, $\tilde{\rho}'$, $\tilde{a}_\cont$, $\tilde{a}_\cont'$, and $\tilde{\rho}''$, there is a step
$(\aexpr, \tilde{\rho}', \tilde{\sigma}, \tilde{\sigma}_\cont, \tilde{a}_\cont') \fsToSub (e, \tilde{\rho}, \tilde{\sigma}, \tilde{\sigma}_\cont, \tilde{a}_\cont)$ and, by the induction hypothesis, a path:
\begin{align*}
  (e_0, \varnothing, \epsilon) \pdPath \hC((\aexpr, \tilde{\rho}', \tilde{a}_\cont'), ((x, e, \tilde{\rho}''), \tilde{a}_\cont):\tilde{\psi})\ \ \text{(via $\hat{r}$, $\hat{\sigma}$)}
\end{align*}
\vspace{-0.5cm}
\begin{align*}
  \text{where}\ \ \ \ \ \ \tilde{\rho} = \tilde{\rho}''[ x \mapsto \widetilde{alloc}(x, (\aexpr, \tilde{\rho}', \tilde{\sigma}, \tilde{\sigma}_\cont, \tilde{a}_\cont')) ]
\end{align*}
We can then show that $(\hat{r}, \hat{\sigma})$ contains a step corresponding to the step in $(\tilde{r}, \tilde{\sigma}, \tilde{\sigma}_\cont)$:
\begin{align*}
  (\aexpr, \hEnv(\tilde{\rho}'), \hat{\sigma}, (x, e, \hEnv(\tilde{\rho}'')):\hKont(\tilde{\psi}))
  \ \ \ \ \ \ \ \ \ \ \ \ &\\
  \fsToSub
  (e, \hEnv(\tilde{\rho}), \hat{\sigma}, \hKont(\tilde{\psi}))&
\end{align*}
Combining this with the path from the induction hypothesis, we can then construct the path in our conclusion.

\textbf{Case:} The last step of the path is a call.
For some $x$, $y$, $f$, $\aexpr$, $e'$, $\tilde{\rho}'$, $\tilde{\rho}_\lambda$, $\tilde{a}_\cont$, $\tilde{\psi}'$, we have
$(y, e, \tilde{\rho}_\lambda) \in \tilde{\mathcal{A}}(f, \tilde{\rho}', \tilde{\sigma})$, and a step:
\vspace{-0.5em}
\begin{align*}
  (\letiform{x}{\appform{f}{\aexpr}}{e'}, \tilde{\rho}', \tilde{\sigma}, \tilde{\sigma}_\cont, \tilde{a}_\cont') \fsToSub (e, \tilde{\rho}, \tilde{\sigma}, \tilde{\sigma}_\cont, \tilde{a}_\cont),\ \text{where}
\end{align*}
\vspace{-0.5cm}
\begin{align*}
  \tilde{\rho} = \tilde{\rho}_\lambda[ x \mapsto \widetilde{alloc}(x, (\letiform{x}{\appform{f}{\aexpr}}{e'}, \tilde{\rho}', \tilde{\sigma}, \tilde{\sigma}_\cont, \tilde{a}_\cont')) ]
\end{align*}
and, by the induction hypothesis, a path:
\begin{align*}
  (e_0, \varnothing, \epsilon) \pdPath \hC((\letiform{x}{\appform{f}{\aexpr}}{e'}, \tilde{\rho}', \tilde{a}_\cont'), \tilde{\psi}')\ \ \text{(via $\hat{r}$, $\hat{\sigma}$)}
\end{align*}
We can then show that $(\hat{r}, \hat{\sigma})$ contains a step corresponding to the step in $(\tilde{r}, \tilde{\sigma}, \tilde{\sigma}_\cont)$:
\begin{align*}
  (\letiform{x}{\appform{f}{\aexpr}}{e'}, \hEnv(\tilde{\rho}'), \hat{\sigma}, \hKont(\tilde{\psi}'))
  \ \ \ \ \ \ \ \ \ \ \ \ &\\ 
  \fsToSub
  (e, \hEnv(\tilde{\rho}), \hat{\sigma}, (x, e', \hEnv(\tilde{\rho}')):\hKont(\tilde{\psi}'))&
\end{align*}
Combining this with the path from the induction hypothesis, we can then construct the path in our conclusion.
\end{proof}

Then, with lemma~\ref{lemma:path_end_point}, we show that the endpoint of any path in $\hat{\xi}$ is in $\hat{\xi}$ (i.e., the bottom edge of figure~\ref{figure:logical-chain}).

\begin{lemma}[Path endpoint]
\label{lemma:path_end_point}
If $(\hat{r}, \hat{\sigma}) \pdTo (\hat{r}, \hat{\sigma})$, then for any path
$\hat{c}_0 \pdPath (e, \hat{\rho}, \hat{\cont})\ \ \text{(via $\hat{r}$, $\hat{\sigma}$)}$,
we have: $(e, \hat{\rho}, \hat{\cont}) \in \hat{r}$.
\end{lemma}
\begin{proof}
Trivial. By induction.
\end{proof}

Finally, with lemmas~\ref{lemma:step_prec_Store} and \ref{lemma:step_prec_R}, we show that precision is preserved by the step relation $\fsTo$
(i.e., the right edge of figure~\ref{figure:logical-chain}).
Then in theorem~\ref{theorem:step_prec_xi}, we show that these are all precise, which is ultimately what we want to prove. 

\begin{lemma}[Preservation of precision for value stores] 
\label{lemma:step_prec_Store}
If $(\hat{r}, \hat{\sigma}) \pdToW (\hat{r}, \hat{\sigma})$, $\wf(\tilde{\xi}, (\tilde{r}, \tilde{\sigma}, \tilde{\sigma}_\cont))$,
$(\tilde{r}, \tilde{\sigma}, \tilde{\sigma}_\cont) \fsToW (\tilde{r}', \tilde{\sigma}', \tilde{\sigma}_\cont')$, and
$(\hat{r}, \hat{\sigma}) \simXi (\tilde{r}, \tilde{\sigma}, \tilde{\sigma}_\cont)$, then $\hat{\sigma} \simStore \tilde{\sigma}'$.
\end{lemma}
\begin{proof}
Omitted for space.
\end{proof}

\begin{lemma}[Preservation of precision for reachable configurations] 
\label{lemma:step_prec_R}
If $(\hat{r}, \hat{\sigma}) \pdToW (\hat{r}, \hat{\sigma})$, $\wf(\tilde{\xi}, (\tilde{r}, \tilde{\sigma}, \tilde{\sigma}_\cont))$,
$(\tilde{r}, \tilde{\sigma}, \tilde{\sigma}_\cont) \fsToW (\tilde{r}', \tilde{\sigma}', \tilde{\sigma}_\cont')$, and
$(\hat{r}, \hat{\sigma}) \simXi (\tilde{r}, \tilde{\sigma}, \tilde{\sigma}_\cont)$, then $\hat{r} \simR \tilde{r}'$.
\end{lemma}
\begin{proof}
If we unfold the definition of $(\simR)$, we must show that for all $(e, \hat{\rho}, \hat{\cont}) \in \hat{r}'$ and 
$\tilde{\psi} \inpsi \tilde{a}_\cont\ \text{(via $\tilde{\sigma}_\cont'$)}$, that $(e, \hEnv(\tilde{\rho}), \hKont(\tilde{\psi})) \in \hat{r}$.
%
%
By lemma~\ref{lemma:stacks_have_paths}, we have:
$$((e_0, \varnothing, \tilde{a}_\text{halt}), \epsilon) \fsPath ((e, \tilde{\rho}, \tilde{a}_\cont), \tilde{\psi})\ \text{(via $\tilde{\xi}$, $(\tilde{r}, \tilde{\sigma}, \tilde{\sigma}_\cont)$)}$$
From this, by lemma~\ref{lemma:path_til_hat}, we have: $$(e_0, \varnothing, \epsilon) \pdPath (e, \hEnv(\tilde{\rho}), \hKont(\tilde{\psi}))\ \text{(via $\hat{r}$, $\hat{\sigma}$)}$$
Finally, by lemma~\ref{lemma:path_end_point}, we have our conclusion.
\end{proof}

\begin{theorem}[Precision of analysis results] 
\label{theorem:step_prec_xi}
If $\tilde{\xi}$ is the result of taking zero or more steps of $(\fsToW)$, starting from $(\{(e_0, \varnothing, \tilde{a}_\text{halt})\}, \bot, \bot)$, 
and $\hat{\xi} \pdToW \hat{\xi}$, then $\hat{\xi} \simXi \tilde{\xi}$.
\end{theorem}
\begin{proof}
By induction over the number of steps, trivial simplifications, unfoldings, and lemmas~\ref{lemma:iterated_wf}, \ref{lemma:step_prec_Store}, and \ref{lemma:step_prec_R}.
\end{proof}

\section{Conclusion}
Traditional control-flow analysis has long suffered from return-flow conflation of values, even when context sensitivity and related techniques keep these
values separate across function calls. 
Recent approaches have made significant progress in addressing this problem.
However, each suffers from serious drawbacks.
PDCFA incurs a substantial development cost and causes a quadratic-factor increase in run-time complexity.
AAC is trivial to implement, but incurs an even worse increase in run-time complexity.
Our approach, however, both is simple to implement and adds no asymptotic cost to run-time complexity.
To accomplish this, we synthesize the lessons learned from PDCFA and AAC to show that the ideal continuation address is simply a function's
polyvariant entry point: its expression and abstract binding environment. 
This introspection on entry points and the corresponding choice of continuation address yields a finite-state analysis whose call transitions are
precisely matched with return transitions at no cost to either run-time or development-time overhead.

\acks{This material is partially based on research sponsored by DARPA under
agreement numbers AFRL FA8750-15-2-0092 and FA8750-12-2-0106 and by
NSF under CAREER grant 1350344.  The U.S. Government is authorized to
reproduce and distribute reprints for Governmental purposes
notwithstanding any copyright notation thereon.}

\balance
\bibliography{paper}{}
\bibliographystyle{abbrvnat}

\end{document}